\documentclass[aps,showpacs,pra,twocolumn,superscriptaddress,10pt]{revtex4-2}
\usepackage{graphicx}
\usepackage{times}
\usepackage{physics}
\usepackage{amssymb}
\usepackage{textcomp}
\usepackage{mathtools}
\usepackage[utf8]{inputenc}
\usepackage[polish,english]{babel}
\usepackage{graphicx}
\usepackage[shortlabels]{enumitem}
\usepackage[T1]{fontenc}
\usepackage{xcolor}
\usepackage{comment}
\usepackage{dcolumn}
\usepackage{bm}
\usepackage{verbatim}
\usepackage{amsthm,amsmath}
\usepackage{dsfont,bbm}
\usepackage{pifont}
\usepackage{orcidlink}

\theoremstyle{definition}
\newtheorem{lemma}{Lemma}
\newtheorem{theorem}{Theorem}
\newtheorem{definition}{Definition}
\newtheorem{remark}{Remark}
\newtheorem{proposition}{Proposition}
\newtheorem{conjecture}{Conjecture}

\usepackage{hyperref}
\hypersetup{
    colorlinks,
    citecolor=blue,
    filecolor=black,
    linkcolor=blue,
    urlcolor=blue
}
\usepackage{cleveref}

\newcommand{\whitestar}{\mathrel{\text{\ding{75}}}}

\let\existstemp\exists
\let\foralltemp\forall
\renewcommand*{\exists}{\existstemp\mkern2mu}
\renewcommand*{\forall}{\foralltemp\mkern2mu}

\begin{document}

\title{Certifying Majorana Fermions with Elegant-Like Bell Inequalities\\ and a New Self-Testing Equivalence}

\author{Patryk Michalski\orcidlink{0009-0009-0305-7356}}
\email{pmichalski@cft.edu.pl}
\affiliation{Center for Theoretical Physics, Polish Academy of Science, al. Lotnik\'{o}w 32/46, 02-668 Warsaw, Poland}
\affiliation{Institute of Theoretical Physics, University of Warsaw, Pasteura 5, 02-093 Warsaw, Poland}
\author{Arturo Konderak\orcidlink{0000-0002-4546-2626}}
\email{akonderak@cft.edu.pl}
\affiliation{Center for Theoretical Physics, Polish Academy of Science, al. Lotnik\'{o}w 32/46, 02-668 Warsaw, Poland}
\author{Wojciech Bruzda\orcidlink{0000-0001-9743-7927}}
\affiliation{Center for Theoretical Physics, Polish Academy of Science, al. Lotnik\'{o}w 32/46, 02-668 Warsaw, Poland}
\author{Remigiusz Augusiak\orcidlink{0000-0003-1154-6132}}
\affiliation{Center for Theoretical Physics, Polish Academy of Science, al. Lotnik\'{o}w 32/46, 02-668 Warsaw, Poland}

\begin{abstract}
Bell inequalities provide a fundamental tool for probing nonlocal correlations, yet their quantum bound, that is, the maximal value attainable through quantum strategies, is rarely accessible analytically. In this work, we introduce a general construction of Bell inequalities for which this bound can be computed exactly. Our framework generalizes both the Clauser--Horne--Shimony--Holt and Gisin’s elegant inequalities, yielding Bell expressions maximally violated by any number of pairwise anticommuting Clifford observables together with the corresponding maximally entangled state. Under suitable assumptions, our inequalities also enable the device-independent certification of Majorana fermions, understood as multiqubit realizations of Clifford algebra generators. Importantly, we identify an additional equivalence that must be incorporated into the definition of self-testing beyond invariance under local isometries and transposition. This equivalence arises from partial transposition applied to the shared state and to the measurements, which in specific cases leaves all observed correlations unchanged.
\end{abstract}
\date{\today}

\maketitle

\section{Introduction}

One of the central distinctions between classical and quantum physics lies in the structure of correlations between distant systems. In classical theories, such correlations can always be attributed to local hidden variable (LHV) models. Quantum mechanics, in contrast, predicts correlations that cannot be reproduced by any LHV description. This phenomenon, which is a consequence of quantum entanglement, reveals the intrinsically nonlocal character of quantum theory and is known as Bell nonlocality.

To formalize this distinction, Bell inequalities were introduced as constraints that any LHV theory must satisfy, yet which can be violated by quantum mechanics~\cite{Bell1964}. Such violations not only reveal the presence of nonlocal correlations, but also constitute a powerful resource for tasks such as quantum cryptography~\cite{Acin2007} and quantum communication complexity~\cite{Buhrman2010,Buhrman2016}. Moreover, the violation of suitably chosen Bell inequalities enables \emph{self-testing}~\cite{Mayers2004,Supic2020}---the strongest form of device-independent certification---which allows one to infer both the underlying quantum state and the performed measurements solely from the observed correlations.

The device-independent approach to quantum certification eliminates the need to trust the internal functioning of the devices, making it particularly relevant for secure quantum technologies. This paradigm has attracted significant attention in recent years, leading to numerous self-testing results for entangled quantum states. Notably, all pure bipartite entangled states~\cite{Coladangelo2017} and all multiqubit entangled states~\cite{BalanzoJuando2024} are now known to be self-testable. Parallel advances have established self-testing for various families of measurements, including Pauli observables~\cite{Bowles2018}, all two-dimensional projective measurements~\cite{Yang2013}, certain $d$-outcome observables~\cite{Sarkar2021}, and constant-size measurements with algebraic structure~\cite{Mancinska2024}. More recent developments extend these results to all real projective measurements~\cite{Chen2024} and, in network scenarios, even to arbitrary quantum measurements~\cite{sarkar2024universalschemeselftestquantum}.

Although Bell inequalities that enable self-testing are of great interest, designing them for specific states and measurements remains a notoriously challenging task (see, however, Ref.~\cite{Barizien2024custombell}). Our work aims to partially fill this gap. We introduce a framework for constructing bipartite correlation Bell inequalities with arbitrary numbers of binary observables, based on a class of rectangular matrices with pairwise orthogonal rows and normalized columns. We refer to these as row-orthogonal, column-normalized (ROCN) matrices, and to the associated Bell expressions as ROCN Bell inequalities. A key feature of this family is that their quantum bound can be computed analytically; moreover, it is in general tighter than the best-known universal upper bound for correlation Bell inequalities derived in Ref.~\cite{Epping2013}. 

We identify conditions under which ROCN Bell inequalities guarantee a quantum advantage over all LHV models. Remarkably, these inequalities are maximally violated by maximally entangled states together with sets of Clifford observables of arbitrary size. In this sense, our construction extends the well-known Clauser--Horne--Shimony--Holt~\cite{Clauser1969} and Gisin’s elegant~\cite{gisin2007bellinequalitiesquestionsanswers} Bell inequalities to arbitrarily large families of binary, pairwise anticommuting observables.

A further contribution of our work is the self-testing of Majorana operators (or fermions), which provide a concrete physical realization of Clifford algebra generators. Majorana operators play a central role across several areas of physics: they arise in quantum field theory~\cite{Quantization} and quantum statistical mechanics~\cite{Araki2003} (see also Ref.~\cite{Bratteli1981}). In quantum computation, they are particularly important due to their connection to the Pauli group and the stabilizer formalism~\cite{Nielsen2012}. For an overview of their broader role in fermionic systems, see Ref.~\cite{Szalay2021}. They also feature prominently in entropic uncertainty relations, achieving maximal violations in scenarios involving more than two measurements~\cite{Wehner2010,Coles2017}. A contextuality-based approach for certifying two-qubit Majorana fermions was previously proposed in Ref.~\cite{Irfan2020}.

Importantly, for an odd number of Majorana operators, we show that self-testing is defined only up to a specific transformation that generalizes the well-known invariance of quantum correlations under complex conjugation in the case of three Pauli matrices. This reveals a previously unrecognized form of equivalence that should be incorporated into the definition of self-testing. In fact, it turns out that apart from local isometries and complex conjugation (cf. Ref.~\cite{Supic2020}), in certain situations, the observed correlations can also be invariant under partial transposition of the state and the measurement.


\section{Preliminaries}

Before getting to the results, let us first establish the terminology and notation that will be consistently employed throughout the manuscript, and recall the most relevant concepts and facts concerning Bell nonlocality and self-testing.

\subsection{Bell nonlocality \& Bell inequalities}
\label{sec:bell_non_locality}
We consider the standard Bell scenario involving two spacelike separated observers, Alice and Bob, who share a quantum state $\rho_\mathrm{AB}$ that acts on some finite-dimensional composite Hilbert space $\mathcal{H}_\mathrm{AB}=\mathcal{H}_\mathrm{A} \otimes \mathcal{H}_\mathrm{B}$. On their respective subsystems, Alice and Bob perform measurements chosen from sets of $m$ and $n$ independent two-outcome observables, denoted respectively as $A_i$ and $B_j$, with $i \in \{1,\ldots,m\}$ and $j \in \{1,\ldots,n\}$. Formally, $A_i$ and $B_j$ are Hermitian operators acting on $\mathcal{H}_\mathrm{A}$ and $\mathcal{H}_\mathrm{B}$, respectively, such that $A_i^2 = B_j^2 = \mathbbm{1}$ (where, with a slight abuse of notation, $\mathbbm{1}$ denotes the identity on the corresponding local Hilbert space). Observe that, due to Naimark's dilation theorem~\cite{Holevo2012}, it is sufficient to consider projective measurements~\cite{Baptista2023}. The outcomes observed by Alice and Bob are denoted $a$ and $b$, where $a,b=\pm 1$. 

The correlations generated when the parties perform many rounds of measurements are described by a collection of probability distributions $\{p(a,b|i,j)\}$, where $p(a,b|i,j)$ denotes the probability that Alice and Bob obtain outcomes $a$ and $b$ upon performing measurements $i$ and $j$. Since we deal here with two-outcome measurements, one can equivalently represent the observed correlations through a collection of expectation values~\cite{Goh2018}:
\begin{equation}
   \vec{c}\equiv  \{\langle A_iB_j\rangle, \langle A_i\rangle, \langle B_j\rangle\},
\end{equation}
which are given by Born rule as
\begin{eqnarray}
    &\langle A_i B_j\rangle = \Tr[\rho_\mathrm{AB} (A_i \otimes B_j)],& \nonumber\\[1ex]
    &\langle A_i \rangle = \Tr[\rho_\mathrm{AB} (A_i \otimes \mathbbm{1})],     \quad\langle B_j \rangle = \Tr[\rho_\mathrm{AB} (\mathbbm{1} \otimes B_j)].&\nonumber\\
\end{eqnarray}
Both collections (of probability distributions or expectation values) are commonly referred to as \emph{correlations} or \emph{behaviour}. 

Let us now consider correlations for which the joint expectation values admit a local hidden-variable (LHV) decomposition of the form
\begin{equation}
   \langle A_i B_j\rangle = \sum_{\lambda} p(\lambda)\,\langle A_i\rangle_{\lambda}\langle B_j\rangle_{\lambda},
\end{equation}
where $\lambda$ denotes the hidden variable (or shared randomness), $p(\lambda)$ is a probability distribution, and 
$-1\leqslant \langle A_i\rangle_{\lambda},\langle B_j\rangle_{\lambda}\leqslant 1$ for any $\lambda$ and any pair $i,j$. The set of all correlations admitting such a decomposition forms a convex polytope, whose vertices correspond to local deterministic correlations satisfying $\langle A_i B_j\rangle = \langle A_i\rangle \langle B_j\rangle$ for every pair $i,j$, with all individual expectation values $\langle A_i\rangle$ and $\langle B_j\rangle$ taking values $\pm 1$ for any $\lambda$ and any pair $i,j$. 

Quantum correlations that cannot be represented in terms of local hidden variable models are termed Bell nonlocal or simply nonlocal; the existence of such correlations gives rise to the phenomenon of Bell nonlocality. One can detect nonlocality in the correlations $\vec{c}$ with the aid of Bell inequalities
\cite{Bell1964}, whose generic form in the considered scenario reads
\begin{equation}
    I\vcentcolon=\sum_{i=1}^m \sum_{j=1}^n h_{ij} \langle A_iB_j\rangle+\sum_{i=1}^{m} f_i\langle A_i\rangle+
    \sum_{j=1}^{n} g_j\langle B_j\rangle\leqslant \beta_C.
\end{equation}
Here, $h_{ij}$, $f_i$ and $g_j$ are some real parameters, while $\beta_C$ is the maximal value of the \textit{Bell functional} $I$ over the polytope of correlations admitting the LHV representation, typically referred to as the classical bound. Violation of a Bell inequality by the correlations $\vec{c}$ implies that these correlations, as well as the shared quantum state, are nonlocal.

In this work, we restrict our attention to correlation Bell inequalities, a subclass of Bell inequalities in which all the coefficients corresponding to local terms $f_i$ and $g_j$ vanish:
\begin{equation}\label{eq:CHSH_ineq}
    I\vcentcolon=\sum_{i=1}^m \sum_{j=1}^n h_{ij} \langle A_iB_j\rangle\leqslant\beta_C.
\end{equation}
Such Bell inequalities have been thoroughly studied in the literature (see, e.g., Refs.~\cite{Epping2013,PhysRevLett.113.210403}).

A seminal example of such a Bell inequality is the famous 
CHSH Bell inequality~\cite{Clauser1969} which corresponds to the simplest possible Bell scenario where each observer chooses to perform one of two possible two-outcome measurements ($m=n=2$). This inequality can be written as
\begin{equation}
    I_{\mathrm{CHSH}}\vcentcolon=\langle A_1B_1\rangle+\langle A_1B_2\rangle+\langle A_2B_1\rangle-\langle A_2B_2\rangle \leqslant2.
\end{equation}
The maximal quantum value of $I_{\mathrm{CHSH}}$ is $\beta_Q^{\mathrm{CHSH}}=2\sqrt{2}$, which is achieved by the maximally entangled state of two qubits 
\begin{equation}
    \ket{\Phi_2}=\frac{1}{\sqrt{2}}(|00\rangle+|11\rangle),
\end{equation}
and the observables $A_1=X$, $A_2=Z$ for Alice, 
and $B_{1}=(X+Z)/\sqrt{2}$, $B_2=(X-Z)/\sqrt{2}$ for Bob. Here, $X$, $Y$ and $Z$ denote the Pauli matrices. It is worth pointing out that the observables of Alice and Bob anticommute, i.e., 
$\{A_1,A_2\}=\{B_1,B_2\}=0$.

Another example of a Bell inequality of the form~\eqref{eq:CHSH_ineq} is the so-called elegant Bell inequality (EBI)~\cite{gisin2007bellinequalitiesquestionsanswers}, originally introduced in Ref.~\cite{BechmannPasquinucci2003}, which corresponds to a scenario where Alice performs three measurements and Bob four. It reads:
\begin{equation}\label{eq:elegant_ineq}
    I_{\mathrm{EBI}}\vcentcolon= \begin{aligned}[t]
        &\langle A_1B_1\rangle+
        \langle A_1B_2\rangle-\langle A_1B_3\rangle-
        \langle A_1B_4\rangle\\
        &+\langle A_2B_1\rangle-
        \langle A_2B_2\rangle+\langle A_2B_3\rangle-
        \langle A_2B_4\rangle\\
        &+\langle A_3B_1\rangle-\langle A_3B_2\rangle-\langle A_3B_3\rangle+\langle A_3B_4\rangle\leqslant4.
    \end{aligned}
\end{equation}
This inequality is also maximally violated by the two-qubit maximally entangled state $\ket{\Phi_2}$ and the observables $A_1=X$, $A_2=Y$, $A_3=Z$ for Alice, and 
\begin{equation}
    \begin{aligned}
        & B_1 = \frac{1}{\sqrt{3}}(X + Y + Z), & \quad
        & B_2 = \frac{1}{\sqrt{3}}(X - Y - Z), \\
        & B_3 = \frac{1}{\sqrt{3}}(-X + Y - Z), & \quad
        & B_4 = \frac{1}{\sqrt{3}}(-X - Y + Z),
    \end{aligned}
\end{equation}
for Bob. The maximal quantum value of $I_{\mathrm{EBI}}$ is known to be $\beta_Q^{\mathrm{EBI}}=4\sqrt{3}$. The elegance of the EBI inequality lies in its maximal violation when Alice’s measurement eigenstates form a complete set of three mutually unbiased bases (MUBs), while Bob’s states can be grouped into two sets of symmetric informationally complete (SIC) elements; see Ref.~\cite{Tavakoli2021} for a generalization to higher dimensions ($d \geqslant 3$). It is worth noting that, as in the case of the CHSH inequality, Alice's optimal observables mutually anticommute. This property extends beyond these examples: the Bell inequalities that we identify and characterize in this work are maximally violated precisely by mutually anticommuting observables that form representations of Clifford algebra generators. 

\subsection{Self-testing}\label{sec:self_testing}

Let us now introduce the concept of self-testing. We consider again the Bell experiment that produces correlations $\vec{c}$, but now assume that both the shared state $\rho_\mathrm{AB}$ and the local observables $\{A_i\}$ and $\{B_j\}$ are unknown. The central question is whether the observed correlations uniquely identify, up to natural equivalences, the underlying quantum realization. More precisely, we aim to certify that the state and observables are (up to certain transformations) equivalent to those of a fixed reference quantum experiment, characterized by a known bipartite state $\ket*{\tilde{\psi}} \in \mathbb{C}^d \otimes \mathbb{C}^d$ and known observables $\{\tilde {A}_i\}$ and $\{\tilde B_j\}$ acting on $\mathbb{C}^d$.

We recall the standard definition of self-testing for both the shared state and the local measurements~\cite{Supic2020}. It accounts for the fact that a quantum strategy can be uniquely identified only up to local isometries on each subsystem.

\begin{definition}[Self-testing]\label{def:self-testing}
    We say that a correlation $\vec{c}$ \emph{self-tests} the reference state $\ket*{\tilde{\psi}} \in \mathbb{C}^d \otimes \mathbb{C}^d$ and the reference observables $\{\tilde{A}_i\}$, $\{\tilde{B}_j\}$ acting on $\mathbb{C}^d$ if, for every state $\ket{\psi}\in\mathcal{H}_\mathrm{A} \otimes \mathcal{H}_\mathrm{B}$ and every observables $\{A_i\}$ and $\{B_j\}$ generating the correlation $\vec{c}$, there exist local isometries
    \begin{equation}\label{eq:isometries_self-testing}
        U_\mathrm{A} : \mathcal{H}_\mathrm{A} \to \mathbb{C}^d \otimes \mathcal{H}_\mathrm{A'}, \quad U_\mathrm{B} : \mathcal{H}_\mathrm{B} \to \mathbb{C}^d   \otimes \mathcal{H}_\mathrm{B'},
    \end{equation}
    for proper ancillary Hilbert spaces $\mathcal{H}_\mathrm{A'}$ and $\mathcal{H}_\mathrm{B'}$, such that
    \begin{equation}\label{eq:self-testing}
        U_\mathrm{A}\otimes U_\mathrm{B} \left( A_i\otimes B_j \ket{\psi}\right) = \left( \tilde A_i \otimes\tilde B_j \ket*{\tilde\psi}\right)\otimes \ket{\xi},
    \end{equation}
    for all $i,j$ and some auxiliary state $\ket{\xi}\in\mathcal{H}_{\mathrm{A}
    '}\otimes\mathcal{H}_{\mathrm{B}'}$.
\end{definition}

Beyond local isometries---which account for local unitaries and dilations---the correlations are also invariant under complex conjugation of the state and the observables~\cite{Mosca2011,Bowles2018}. This invariance becomes particularly relevant in the multipartite scenario~\cite{Supic2023,BalanzoJuando2024}. A general definition of self-testing that accommodates complex observables has been formulated in Ref.~\cite{Supic2020}. However, as we show in this work, this definition does not capture the self-testing of Clifford observables: there exists a more general transformation, extending beyond ordinary complex conjugation, that likewise preserves all observed correlations.

To provide an illustrative example of self-testing, let us consider the CHSH Bell inequality and let us assume that it is maximally violated by the state $\ket{\psi}$ and the observables $\{A_i\}$ and $\{B_j\}$. Then, one can prove that there exist unitary operators $U_\mathrm{A}$ for Alice and $U_\mathrm{B}$ for Bob such that
\begin{equation}
    U_\mathrm{A}\otimes U_\mathrm{B}\ket{\psi}=\ket{\Phi_2}\otimes\ket{\xi},
\end{equation}
where $\ket{\Phi_2}$ is the maximally entangled state of two qubits, and

\begin{eqnarray}
    &U_\mathrm{A}\, A_1\, U_\mathrm{A}^{\dagger}=X\otimes\mathbbm{1},\quad  U_\mathrm{A}\, A_2\, U_\mathrm{A}^{\dagger}=Z\otimes\mathbbm{1},&\nonumber \\ 
    &U_\mathrm{B}\, B_j \, U_\mathrm{B}^{\dagger}=\displaystyle\frac{1}{\sqrt{2}}[X+(-1)^{j+1}Z]\otimes\mathbbm{1}.&
\end{eqnarray}

An analogous self-testing statement can be established for the elegant Bell inequality presented in Ref.~\cite{PhysRevA.96.032119}. In this case, the observables to be self-tested are not real, and the alternative definition must be imposed~\cite{Bowles2018}.

\subsection{Clifford algebra}\label{sec:clifford_algebra}

Clifford algebras were originally introduced as a generalization of Grassmann’s exterior algebras~\cite{Clifford1878,Lounesto2001,Lundholm2009}, and today they play an important role in both physics and mathematics~\cite{Quantization,Derezinski2023}. In quantum information theory and the foundations of quantum mechanics, they arise in a variety of contexts---for instance, in formulations of entropic uncertainty relations~\cite{Wehner2010} and in steering inequalities exhibiting unbounded violations~\cite{Marciniak2015}. Although several equivalent definitions of Clifford algebras exist~\cite{Lounesto2001}, for our purposes we focus on the associative algebra generated by a collection of Hermitian bounded operators $\{A_i\}_{i=1}^m$ satisfying the canonical anticommutation relations~\cite{Derezinski2023}:
\begin{equation}\label{eq:anticomm_th1}
     \forall i,k \in \{1,\dots,m\}: \quad \{A_i,A_k\}=2\delta_{ik}\,\mathbbm{1}.
\end{equation}
In the following, we refer to the operators $A_i$ as \textit{Clifford observables} or \textit{generators}. A fundamental result, known as the Jordan--Wigner representation~\cite{Jordan1928,Samoilenko1991}, provides an explicit realization of these operators as acting on a system of qubits, usually referred to as Majorana fermions.
\begin{theorem}[Jordan--Wigner Representation Theorem]\label{th:1}
    Let $\{A_i\}_{i=1}^m$ be a set of observables acting on a finite-dimensional Hilbert space $\mathcal{H}$ and satisfying the canonical anticommutation relations~\eqref{eq:anticomm_th1}. Let $m=2r+\varepsilon$, with $\varepsilon\in\{0,1\}$ and $r=\lfloor m/2\rfloor$. Then, up to a unitary transformation, the Hilbert space $\mathcal{H}$ can be decomposed as a tensor product of qubits:
    \begin{equation}\label{eq:tensor_product_JW}
        \mathcal{H} \cong \bigotimes_{k=1}^{r} \mathbb{C}^2 \otimes \mathcal{H}'.
    \end{equation}
    With respect to this decomposition, the first $2r$ operators $A_i$ can be expressed in terms of Pauli matrices as    
    \begin{equation}\label{eq:U}
        A_{i} =
        \begin{cases}
            Y^{\otimes (i-1)/2} \otimes Z \otimes \mathbbm{1}^{\otimes (2r-i-1)/2}_2 \otimes \mathbbm{1}_{\mathcal{H}'} & \text{for } i \text{ odd}, \\
            Y^{\otimes (i-2)/2} \otimes X \otimes \mathbbm{1}^{\otimes(2r-i)/2}_2 \otimes \mathbbm{1}_{\mathcal{H}'} & \text{for } i \text{ even}.
        \end{cases}
    \end{equation}
    For $\varepsilon=1$ (i.e., odd $m$), the last operator takes the form
    \begin{equation}\label{eq:last_operator_JW}
        A_m = Y^{\otimes r} \otimes A_m',
    \end{equation}
    where $A_m'$ is a Hermitian and unitary operator acting on $\mathcal{H'}$.
\end{theorem}
For completeness, the proof of the above theorem is presented in Appendix~\ref{app:ST-theorems}.

We singled out the last operator in Eq.~\eqref{eq:last_operator_JW} because it is not uniquely defined. This is directly related to the existence of two unitarily inequivalent irreducible representations of the Clifford algebra in the odd case. Both representations act on the Hilbert space

\begin{equation}
    \mathcal{H}=\bigotimes_{k=1}^r \mathbb{C}^2,
\end{equation}
and can be transformed into one another by the action of the following outer automorphism~\cite{Derezinski2023}:
\begin{equation}\label{eq:automporhism_alpha}
\alpha(A_{i_1} \ldots A_{i_k}) = (-1)^k A_{i_1} \ldots A_{i_k},
\end{equation}
followed by the unitary transformation $A_i\mapsto Y^{\otimes r}A_i Y^{\otimes r}$, where the generators $A_i$ are given by Eqs.~\eqref{eq:U} and~\eqref{eq:last_operator_JW}, with $\mathcal{H}'=\mathbb{C}$. The automorphism~\eqref{eq:automporhism_alpha} acts uniformly on all generators, effectively multiplying each of them by $-1$. For an odd number of generators, this implies that the last generator $A_m$ cannot be mapped back to itself by the chosen unitary transformation. Consequently, Clifford observables with odd $m$ can be self-tested only up to an additional operator $A'_m = P_0-P_1$, where $P_0$ and $P_1$ are orthogonal projections satisfying $P_0+P_1=\mathbbm{1}_{\mathcal{H}'}$. The two subspaces defined by $P_0$ and $P_1$ correspond to the two inequivalent irreducible representations of the Clifford algebra.

\section{Bell inequalities from ROCN matrices}\label{sec:ROCN_definition}

In this section, we introduce a general framework for constructing bipartite Bell inequalities involving arbitrary numbers of binary observables, which attain their maximal quantum violation for sets of Clifford observables acting on maximally entangled states. Our construction falls within the class of correlation Bell inequalities and recovers several well-known examples, including the CHSH inequality~\cite{Clauser1969} and Gisin’s elegant Bell inequality~\cite{gisin2007bellinequalitiesquestionsanswers}. Moreover, it exhibits structural similarities with the family of Platonic Bell inequalities introduced in Ref.~\cite{Pal2022}, as discussed in Subection~\ref{sec:comparison_literature}. Importantly, in suitable settings, the maximal violation of our inequalities enables self-testing of Clifford observables.

We begin with the definition of an ROCN matrix.

\begin{definition}[ROCN matrix]\label{def:ROCN_matrix}
Let $h$ be a real $m\times n$ matrix with $m\leqslant n$. We say that $h$ is a \emph{row-orthogonal and column-normalized (ROCN) matrix} if all its rows are nonzero and it satisfies the following conditions:
\begin{widetext}
\begin{align}
    &\forall i, k \in \{1,\ldots,m\}: \quad \sum_{j = 1}^n h_{ij} h_{kj} = \delta_{ik} \sum_{j = 1}^n \label{eq:hmatrix_1}h_{ij}^2 \quad \text{--- orthogonality of rows}, \\ \label{eq:hmatrix_2}
    &\forall j \in \{1,\ldots,n\}: \quad \quad \sum_{i = 1}^m h_{ij}^2 = 1 \quad \text{--- normalization of columns}.
\end{align}
\end{widetext}
\end{definition}

As we will see in Subsection~\ref{sec:quantum-bound} and Section~\ref{sec:exact-self-testing}, the conditions defining an ROCN matrix play a central role in deriving the quantum bound, establishing that the resulting Bell inequalities are nontrivial, and guaranteeing the possibility of self-testing. We now turn to the definition of the Bell observable, which encodes a Bell inequality in operator form: its expectation value on a given quantum state reproduces the corresponding Bell expression.

Given an ROCN matrix $h\in\mathbb{R}^{m \times n}$, we define the corresponding \textit{ROCN Bell operator} as
\begin{equation}\label{eq:bell_operator}
    {\mathcal{B}}_{h} = \sum_{i = 1}^m \sum_{j = 1}^n h_{ij} A_i \otimes B_j,
\end{equation}
where $\{A_i\}_{i=1}^m$ and $\{B_j\}_{j=1}^n$ are the binary observables of Alice and Bob, respectively. The associated ROCN Bell inequality is then
\begin{equation}
    I_h = \sum_{i = 1}^m \sum_{j = 1}^n h_{ij} \langle A_i  B_j\rangle \leqslant \beta_C^{h},
\end{equation}
where $\beta_C^h$ denotes the classical (local) bound. In the following, we refer to $I_h$ as the \textit{ROCN Bell functional}. 

It is worth pointing out that both the CHSH and Gisin’s elegant Bell inequalities can be expressed in the form~\eqref{eq:bell_operator} after a suitable rescaling. We revisit this point in Subsection~\ref{sec:comparison_literature}, where we systematically compare our framework with other existing Bell inequalities. 

\subsection{Quantum bound}
\label{sec:quantum-bound}

We now determine the maximal quantum value $\beta_Q^h$ of the Bell operator $ {\mathcal{B}}_h$ defined in Eq.~\eqref{eq:bell_operator}. Our approach relies on a sum-of-squares (SOS) decomposition~\cite{Supic2020} that allows us to both upper-bound the operator and construct a quantum realization that saturates the bound.

Specifically, we seek a family of operators $\{N_\alpha\}$, depending on the observables $A_i$ and $B_j$, such that
\begin{equation}\label{eq:sos_decomposition}
    \beta_Q^h \mathbbm{1} - {\mathcal{B}}_h = \frac{1}{2} \sum_{\alpha} N_\alpha^\dagger N_\alpha \geqslant 0.
\end{equation}
Once such a decomposition is established, one can try to identify a physical realization achieving $\beta_Q^h$ by finding a quantum state $\ket{\psi}$ that lies in the kernel of each $N_\alpha$, i.e., $N_\alpha\ket{\psi}=0$ for all $\alpha$. Thanks to the specific structure of the ROCN Bell operator~\eqref{eq:bell_operator}, this procedure is particularly straightforward and effective.

\begin{proposition}[Quantum bound]\label{prop:quantum_bound}
    Let $h$ be an $m\times n$ ROCN matrix. Then, the maximal quantum value of the corresponding ROCN Bell functional $I_h$ is $\beta_Q^h = n$. 
\end{proposition}
Before presenting the proof, we introduce the strategy that achieves the maximal quantum value. Let $m=2r+\varepsilon$ as in Theorem~\ref{th:1}. The quantum bound $\beta_Q^h=n$ is attained when Alice and Bob have access to Hilbert spaces
\begin{equation}
    \mathcal{H}_\mathrm{A}=\mathcal{H}_\mathrm{B}=\bigotimes_{k=1}^r \mathbb{C}^2,
\end{equation}
and share the maximally entangled state
\begin{equation}\label{eq:me_state}
    \ket{ \Phi_d}=\frac{1}{\sqrt d}\sum_{\boldsymbol{i}\in\{0,1\}^r}\ket{\boldsymbol{i}}_{\mathrm{A}}\otimes\ket{\boldsymbol{i}}_{\mathrm{B}}\in\mathcal{H}_\mathrm{A}\otimes \mathcal{H}_\mathrm{B},
\end{equation}
where $d=2^r$, $\boldsymbol{i}=(i_1,\dots,i_r)$, and $\ket{\boldsymbol{i}}=\ket{i_1} \otimes \ket{i_2} \otimes \dots \otimes \ket{i_r}$ is the canonical tensor product basis. Alice’s optimal observables $\{\tilde A_i\}_{i=1}^m$ form a set of mutually anticommuting operators:
\begin{equation}\label{eq:alice_observables}
    \forall i,k\in\{1,\dots,m\}: \quad \{\tilde A_i,\tilde A_k\}=2\delta_{ik}\mathbbm{1},
\end{equation}
and are explicitly constructed according to Eqs.~\eqref{eq:U} and~\eqref{eq:last_operator_JW}, with $\mathcal{H}'=\mathbb{C}$ and, for odd $m$, $A'_m = {1}$. Bob’s observables $\{\tilde B_j\}_{j=1}^n$ are defined as
\begin{equation}\label{eq:bob_observables}
   \forall\, j\in\{1,\dots,n\}: \quad \tilde B_j=\sum_{i=1}^m h_{ij}\, \tilde A_i^\mathsf{T},
\end{equation}
where $\tilde A_i^\mathsf{T}$ denotes the transpose of $\tilde A_i$ with respect to the canonical tensor product basis of $\mathcal{H}_{\mathrm{B}}$.
\begin{definition}\label{def:reference_strategy}
    The quantum strategy specified by the state~\eqref{eq:me_state} together with the measurements~\eqref{eq:alice_observables} and~\eqref{eq:bob_observables} is called the \textit{reference} or \textit{reference strategy} for achieving the maximal quantum value $\beta_Q^{h}$ of the ROCN Bell functional $I_h$.
\end{definition}

\begin{proof}[Proof of Proposition~\ref{prop:quantum_bound}.]
We begin by upper-bounding the Bell operator via a suitable SOS decomposition. For each $j\in\{1,\dots,n\}$, define the operators
\begin{equation}
   N_j \vcentcolon= \mathbbm{1}\otimes \mathbbm{1} -  \sum_{i = 1}^m h_{ij} A_i \otimes B_j.
\end{equation}
A direct computation, using the ROCN conditions~\eqref{eq:hmatrix_1}--\eqref{eq:hmatrix_2}, gives
\begin{align}\label{eq:sos}
    \frac{1}{2} \sum_{j=1}^n N_j^\dag N_j 
    &= \frac{1}{2} \sum_{j = 1}^n \left(\mathbbm{1} \otimes \mathbbm{1} -  \sum_{i = 1}^m h_{ij} A_i \otimes B_j \right)^2 
    \nonumber\\
    &= n\,\mathbbm{1}\otimes\mathbbm{1} -  {\mathcal{B}}_h \geqslant 0.
\end{align}
This provides an upper bound on the maximal quantum value:
\begin{equation}\label{eq:B_quantum_bound}
    \beta_Q^h = \max_{\ket{\psi},\{A_i\},\{B_j\}} \bra{\psi}  {\mathcal{B}}_h \ket{\psi} \leqslant n.
\end{equation}
Inequality~\eqref{eq:B_quantum_bound} is saturated precisely when there exists a state $\ket{\psi}$ lying in the kernel of all $N_j$. Equivalently, this requires
\begin{equation}\label{eq:saturation_quantum_bound}
\forall j \in \{1,\ldots,n\}: \quad \left(\sum_{i = 1}^m h_{ij} A_i \otimes B_j \right) \ket{\psi} = \ket{\psi}.
\end{equation}

To verify that the bound~\eqref{eq:B_quantum_bound} is tight, consider the reference strategy given in Definition~\ref{def:reference_strategy}. Bob’s observables $\{\tilde B_j\}_{j=1}^n$ are Hermitian by construction and satisfy
\begin{align}
    \tilde B_j^2 &= \bigg(\sum_{i=1}^m h_{ij} \tilde A_i^\mathsf{T} \bigg)^2 \nonumber\\ 
    &= \sum_{i < k}^m h_{ij} h_{kj}\{\tilde A_i,\tilde A_k\} + \sum_{i=1}^m h_{ij}^2 \,\mathbbm{1} = \mathbbm{1},
\end{align}
where we used that the observables $\{\tilde A_i\}$ mutually anticommute and that $h$ has column-normalized entries. Next, we employ the well-known flip identity
\begin{equation}\label{eq:flip-flop}
    \mathcal{O}\otimes \mathbbm{1} \,\ket{\Phi_d} = \mathbbm{1} \otimes \mathcal{O}^\mathsf{T}\ket{\Phi_d},
\end{equation}
valid for any operator $\mathcal{O} \in\mathcal B(\mathcal{H}_\mathrm{A})$. Applying this relation, we obtain
\begin{align}\label{eq:kernel_N}
    \left(\sum_{i = 1}^m h_{ij}\tilde A_i \otimes\tilde B_j \right) \ket{\Phi_d} 
    &= (\tilde B_j^\mathsf{T}\otimes\tilde B_j)\ket{\Phi_d} \nonumber\\
    &= (\mathbbm{1}\otimes\tilde B_j^2)\ket{\Phi_d} 
    = \ket{\Phi_d}.
\end{align}
Thus, $\ket{\Phi_d}$ lies in the kernel of every $N_j$, meaning that the SOS bound~\eqref{eq:B_quantum_bound} is saturated. We therefore conclude that $\beta_Q^h=n$.
\end{proof}

\subsection{Classical bound}
\label{sec:classical-bound}

Having established that the quantum bound is given by $\beta_Q^h = n$, we now turn our attention to determining the maximal classical value $\beta_C^h$. We provide sufficient conditions under which a non-trivial gap $\beta_Q^h > \beta_C^h$ exists. Such a gap implies a genuine quantum advantage, and its existence imposes further constraints on the structure of the ROCN matrix.

As discussed in Subsection~\ref{sec:bell_non_locality}, in a classical and realistic description, the outcomes of observables are assumed to be predetermined by a local hidden variable $\lambda$, distributed according to a probability distribution $p(\lambda)$, prior to any measurement being performed~\cite{Supic2020,Augusiak2014}. For a Bell operator generated by the ROCN matrix $h$, the classical expectation value takes the form
\begin{equation}
    \langle  {\mathcal{B}}_h \rangle = \sum_{\lambda} p(\lambda) \sum_{i=1}^m \sum_{j=1}^n h_{ij} \langle{A_i}\rangle_\lambda \langle B_j\rangle_\lambda.
\end{equation}
Since the Bell operator is linear in the expectation values $\langle A_i\rangle_\lambda$ and $\langle B_j\rangle_\lambda$, which are fixed for each $\lambda$, the classical bound $\beta^h_C$ is attained by considering only deterministic assignments of outcomes, i.e., the extremal points of the local polytope:
\begin{equation}
    \beta_C^h = \max_{\substack{\vec{a} \in \{-1,1\}^m \\ \vec{b} \in \{-1,1\}^n}} \sum_{j = 1}^n \sum_{i = 1}^m h_{ij} a_i b_j.
\end{equation}

This expression can be simplified by noting that the optimal choice of $\vec{b}$ for a fixed $\vec{a}$ is given by $b_j = \mathrm{sgn}(\sum_{i=1}^m h_{ij} a_i)$, leading to
\begin{equation}
\label{reduced-trick}
    \beta^h_C = \max_{\vec{a} \in \{-1,1\}^m} \sum_{j = 1}^n \left| \sum_{i = 1}^m h_{ij} a_i \right|.
\end{equation}
Using the normalization condition on the columns of $h$, we can further express the bound as
\begin{equation}
    \beta_C^h = \max_{\vec{a} \in \{-1,1\}^m} \sum_{j = 1}^n \sqrt{1 + 2 \sum_{i < k}^m h_{ij} h_{kj} a_i a_k }.
\end{equation}
Applying the inequality $\sqrt{1+x} \leqslant 1 + \frac{1}{2} x$, which holds for $x \in [-1, \infty[$, we obtain the following upper bound on the value of $\beta_C^h$:
\begin{equation}\label{eq:beta_c}
    \beta_C^h \leqslant \max_{\vec{a} \in \{-1,1\}^m} \sum_{j = 1}^n \left(1 + \sum_{i < k}^m h_{ij} h_{kj} a_i a_k \right) = n.
\end{equation}
This inequality becomes tight (that is, $\beta_C^h = n$) if there exists an assignment $\vec{a} \in \{-1,1\}^m$ for which the pairwise contributions vanish for all columns of $h$, meaning
\begin{equation}\label{eq:triviality}
    \forall j \in \{1,\ldots,n\}: \quad \sum_{i < k}^m h_{ij} h_{kj} a_i a_k = 0.
\end{equation}
Therefore, in order to ensure a {\it quantum violation} of the Bell inequality (that is, to have $\beta_C^h < \beta_Q^h = n$), it is necessary and sufficient to impose the following condition on the matrix $h$:
\begin{equation}\label{eq:non-triviality}
    \forall \vec{a} \in \{-1,1\}^m \; \exists j \in \{1,\ldots,n\}: \quad \sum_{i < k}^m h_{ij} h_{kj} a_i a_k \ne 0.
\end{equation}
While it is simple to construct examples where the bound in Eq.~\eqref{eq:beta_c} is not saturated, in Subsection~\ref{sec:hadamard_matrices} we will return to the question of whether ROCN matrices exist for which $\beta_Q^h = \beta_C^h$.

\section{Self-testing from ROCN Bell inequalities}
\label{sec:exact-self-testing}

As demonstrated in the previous sections, ROCN matrices can be employed to construct non-trivial Bell inequalities that are maximally violated by Clifford observables acting on maximally entangled states. We now proceed to identify the further conditions required to obtain a self-testing statement. Specifically, our goal is to show that the maximal quantum violation of the ROCN Bell inequality not only signals the presence of non-classical correlations, but also certifies---up to well-defined equivalences---the form of the measurements and the shared quantum state.

In the proof of Proposition~\ref{prop:quantum_bound}, we introduced an SOS decomposition of $ {\mathcal B}_h$ [see Eq.~\eqref{eq:sos}]. As will become clear in the proof of Theorem~\ref{th:self_testing_statement}, the saturation condition~\eqref{eq:saturation_quantum_bound} can be rewritten as a homogeneous linear system in the anticommutators $\{A_i, A_k\}$:
\begin{equation}\label{eq:eqnset}
    \forall j \in \{1,\ldots,n\}: \quad \sum_{i < k}^m h_{ij}h_{kj} \{A_i, A_k\} = 0.
\end{equation}
This linear system is characterized by an $n \times \tfrac{1}{2}m(m-1)$ matrix $M$. If $M$ has full column rank, it follows that the observables ${A_i}$ form a family of pairwise anticommuting operators. By the Jordan--Wigner representation theorem, Alice’s observables can then be brought, via a suitable unitary transformation, to the canonical form given in Definition~\ref{def:reference_strategy}. However, when the number of Alice’s observables is odd, the Clifford algebra admits two unitarily inequivalent irreducible representations (see Subsection~\ref{sec:clifford_algebra}). Consequently, there exist two quantum strategies that both reach the quantum bound, yet are not related by any dilation or local unitary transformation (local isometry).

As a result, under the standard notion of self-testing given in Definition~\ref{def:self-testing}, the reference strategy cannot be self-tested for odd $m$. A similar issue was addressed in Ref.~\cite{Kaniewski17} for self-testing of a two-qubit system based on commutation relations, and in Ref.~\cite{PhysRevA.96.032119} for self-testing of the elegant Bell inequality. Nevertheless, one can prove that self-testing exists up to a specific transformation introduced in Subsection~\ref{sec:clifford_algebra}. Specifically, let us consider the reference strategy from Definition~\ref{def:reference_strategy}, with odd $m=2r+1$, and define the transformed observables $\{\tilde A_i^{\whitestar}\}$, $\{\tilde B_j^{\whitestar}\}$ as
\begin{align}\label{eq:A_whitestar}
   \forall i\in\{1, \dots,2r\}:&\quad \tilde A_i^{\whitestar}=\tilde A_i,\quad \tilde A_{2r+1}^{\whitestar}= -\tilde A_{2r+1},\\
   \forall j\in\{1,\dots, n\}:&\quad  \tilde B_j^{\whitestar}= \sum_{i=1}^m h_{ij} \tilde A_i^{\whitestar\mathsf{T}}.\label{eq:B_whitestar}
\end{align}

\begin{remark}\label{remark:partial_transposition}
    Notice that $\tilde A_{i}^{\whitestar}= Y^{\otimes r} \alpha(\tilde A_i)Y^{\otimes r}$, where $\alpha$ is defined in Eq.~\eqref{eq:automporhism_alpha}.
    Alternatively, we can obtain both $A_i^{\whitestar}$ and $B_j^{\whitestar}$ as partial transposition on the last qubit. Explicitly, define
    \begin{equation}\label{eq:partial_trace}
        \mathsf{T}_r : \mathcal O_1\otimes \dots \otimes \mathcal O_r\mapsto \mathcal O_1\otimes  \dots \otimes \mathcal O_r^\mathsf{T},
    \end{equation}
    for any $\mathcal{O}_i \in \mathcal{B}(\mathbb{C}^2)$. It is straightforward to verify that $\tilde A_i^{\whitestar}=\mathsf{T}_r(\tilde A_i)$ and $\tilde B_j^{\whitestar}=\mathsf{T}_r(\tilde B_j)$. Furthermore, the state $\ket{\Phi_d}$ transforms as
    \begin{align}
        \mathsf T_r \otimes \mathsf T_r (\ketbra{\Phi_d})&=\ketbra{\Phi_2}\otimes \dots \otimes (\ketbra{\Phi_2})^{\mathsf T}\nonumber\\
        &= \ketbra{\Phi_d},
    \end{align}
    and is therefore invariant. It is important to emphasize that the partial transposition $\mathsf T_r$ is not a positive map and thus does not preserve positivity of arbitrary quantum states. The invariance holds here only because $\ket{\Phi_d}$ factorizes into a tensor product of $r$ two-qubit maximally entangled states, each of which is invariant under transposition on both subsystems. Since partial transposition preserves the trace, the observed correlations remain unchanged.
\end{remark}

In the ROCN framework, we say that the correlation $\vec c$ self-tests the reference strategy if, for every state $\ket{\psi}\in\mathcal{H}_{\mathrm{A}}\otimes \mathcal{H}_{\mathrm{B}}$ and every observables $\{A_i\}$ and $\{B_j\}$ generating $\vec c$, there exist local isometries
\begin{align}
    U_{\mathrm{A}}:\mathcal{H}_\mathrm{A} \rightarrow \mathbb{C}^d \otimes \mathcal{H}_{\mathrm{A}'}\otimes \mathcal{H}_{\mathrm{A}''}, \\ U_{\mathrm{B}}:\mathcal{H}_\mathrm{B} \rightarrow \mathbb{C}^d \otimes \mathcal{H}_{\mathrm{B}'}\otimes \mathcal{H}_{\mathrm{B}''},
\end{align}
for proper ancillary Hilbert spaces $\mathcal{H}_{\mathrm{A}'}$, $\mathcal{H}_{\mathrm{A}''}$, $\mathcal{H}_{\mathrm{B}'}$ and $\mathcal{H}_{\mathrm{B}''}$ such that
\begin{align}\label{eq:two_parties_self-testing}
    U_{\mathrm{A}}\otimes U_{\mathrm{B}} \left(A_i\otimes B_j\ket{\psi}\right) = \left( \bar A_i \otimes \bar B_j \right) \ket*{\tilde\psi} \otimes \ket{\xi},
\end{align}
for all $i,j$ and some auxiliary state $\ket{\xi} \in \mathcal{H}_{\mathrm{A}'} \otimes \mathcal{H}_{\mathrm{A}''} \otimes \mathcal{H}_{\mathrm{B}'} \otimes \mathcal{H}_{\mathrm{B}''} $. The operators $\bar A_i$ and $\bar B_j$ take the form
\begin{align}
    &\bar{A}_i = \tilde{A}_i \otimes M_0 + \tilde{A}_i^{\whitestar}\otimes M_1, \\
    &\bar{B}_j = \tilde{B}_j \otimes N_0 + \tilde{B}_j^{\whitestar} \otimes N_1,
\end{align}
where ${M_0,M_1}$ and ${N_0,N_1}$ are POVMs acting on $\mathcal{H}_{\mathrm{A}'}$ and $\mathcal{H}_{\mathrm{B}'}$, respectively, and satisfying
\begin{equation}
    \mel{\xi}{(M_0\otimes N_0+M_1\otimes N_1)}{\xi}=1.
\end{equation}

This construction closely resembles the notion of self-testing for complex measurements~\cite[Definition 6]{Supic2020}. As we explain in Appendix~\ref{app:partial_transposition_vs_cc_u}, the two definitions coincide only when $m \equiv 3\,(\text{mod } 4)$, corresponding to an odd number of qubits in the Jordan--Wigner representation. In this case, the partial transposition of the reference observables is unitarily equivalent to complex conjugation. By contrast, when $m \equiv 1\,(\text{mod } 4)$, corresponding to an even number of qubits, this equivalence breaks down, and the definition in Ref.~\cite{Supic2020} no longer captures all transformations that leave the observed correlations invariant. This shows the existence of new equivalences that should be taken into account in the definition of self-testing.

With these observations in place, we now proceed to outline the general characterization of self-testing for ROCN Bell inequalities.

\begin{theorem}[Self-testing from ROCN Bell inequalities]\label{th:self_testing_statement}
    Let $h$ be an $m \times n$ ROCN matrix, and define the $n \times \tfrac{1}{2}m(m-1)$ matrix $M$ as
    \begin{equation}\label{eq:matrix_M}
        M_{j,(i,k)} = h_{ij} h_{kj},
    \end{equation}
    for $j \in \{1,\dots,n\}$ and $i, k \in \{1,\dots,m\}$ with $i < k$, where each pair $(i,k)$ is treated as a single index. 
    The maximal violation of the corresponding ROCN Bell inequality self-tests the reference strategy introduced in Definition~\ref{def:reference_strategy} if and only if $M$ has full column rank:
    \begin{equation}\label{eq:necessary_self_testing}
            \rank(M) = \frac{1}{2}m(m-1).
    \end{equation}
    For odd $m$, the reference strategy is self-tested in the sense of Eq.~\eqref{eq:two_parties_self-testing}.
\end{theorem}

The detailed proof of Theorem~\ref{th:self_testing_statement} is presented in Appendix~\ref{app:proof_main_theorem}.

Note that in the case of even $m$, although the observables in Eq.~\eqref{eq:U} are generally complex, they remain unitarily equivalent to their complex conjugates. In this sense, self-testing is still possible under the standard Definition~\ref{def:self-testing}. 

\section{Applications}\label{sec:applications}

\subsection{Comparison with other families of Bell inequalities}
\label{sec:comparison_literature}

In the previous sections, we demonstrated that it is possible to construct a special class of Bell inequalities generated by ROCN matrices. Under certain conditions, the maximal violation of such an inequality allows one to self-test the reference strategy introduced in Definition~\ref{def:reference_strategy}. Notably, both the CHSH and Gisin's elegant Bell inequalities, given in Eqs.~\eqref{eq:CHSH_ineq} and~\eqref{eq:elegant_ineq}, can be recovered within this framework by choosing their corresponding ROCN matrices as
\begin{gather}
    h_{\mathrm{CHSH}}=\frac{1}{\sqrt 2}
    \begin{bmatrix}
         +1 & +1 \\
         +1 & -1
    \end{bmatrix}, \\
    h_{\mathrm{EBI}}=\frac{1}{\sqrt 3}
    \begin{bmatrix}
         +1 & +1 & -1 & -1 \\
         +1 & -1 & +1 & -1 \\
         +1 & -1 &  -1 &  +1 
    \end{bmatrix}.
\end{gather}
Observe that, up to normalization, the first matrix takes the form of a Hadamard matrix~\cite{Goyeneche2023}, while the second corresponds to a truncated (or partial) Hadamard matrix, as discussed in Subsection~\ref{sec:hadamard_matrices}. It is worth noting that a similar generalization of these two Bell inequalities was proposed in Ref.~\cite{Tavakoli2020} in terms of Platonic solids, and later extended to a more general framework in Ref.~\cite{Pal2022}. A Platonic Bell inequality is defined by considering a matrix $h$ satisfying the ROCN conditions and, in addition,
\begin{equation}\label{eq:hmatrix_3}
    \forall i\in\{1,\dots, m\}:\quad \sum_{j=1}^n h_{ij}^2= \frac{m}{n}.
\end{equation}
In this case, the matrix $h$ is semi-orthogonal~\cite{Abadir2005}. The Bell functional is then constructed from two such semi-orthogonal matrices, $h$ and $h'$, of dimensions $m \times n$ and $m \times n'$, respectively, by defining their Gram product:
\begin{equation}
\forall j\in\{1,\dots, n\},\; k\in\{1,\dots, n'\}:\quad	M_{jk}=\sum_{i=1}^mh_{ij}h'_{ik}.
\end{equation}
The corresponding Bell operator takes the form
\begin{equation}\label{eq:platonic_bell_operator}
	 {\mathcal B}_{h,h'}=\sum_{j=1}^{n}\sum_{k=1}^{n'} M_{jk} A_j\otimes B_k,
\end{equation}
where $\{A_j\}_{j=1}^{n}$ and $\{B_k\}_{k=1}^{n'}$ denote the measurement observables for the two parties. The maximal quantum value of this operator is given by $\beta_{Q}^{h,h'}=nn'/m$. We refer to any Bell inequality of the form~\eqref{eq:platonic_bell_operator} as a \emph{Platonic Bell inequality}. { The name follows from the fact that, when the columns of $h$ correspond to vectors pointing toward the vertices of a Platonic solid, the ROCN conditions together with Eq.~\eqref{eq:hmatrix_3} are satisfied simultaneously.}

It is straightforward to verify that, by setting $n' = m$ and $h' = \delta_{ij}$, the Platonic Bell operator~\eqref{eq:platonic_bell_operator} reduces to an ROCN Bell operator associated with the matrix $h$. A natural question, therefore, is whether the converse holds---namely, whether every ROCN matrix is also Platonic (semi-orthogonal). To address this, we seek explicit examples of ROCN Bell inequalities that are not Platonic, or equivalently, ROCN matrices that do not satisfy the semi-orthogonality condition~\eqref{eq:hmatrix_3}. In what follows, we provide a constructive discussion outlining an algorithm for generating arbitrary ROCN matrices, which can be used to identify such non-Platonic instances.

First, observe that any ROCN matrix $h$ admits a singular value (or polar) decomposition of the form
\begin{equation}\label{eq:singular_value}
    h = S W,
\end{equation}
where $W$ is an $n\times n$ orthogonal matrix, and $S$ is a rectangular diagonal matrix containing the row norms of $h$. Explicitly, defining
\begin{equation}
\forall i \in \{1,\dots,m\}: \quad H_i^2 \coloneqq \sum_{j=1}^{n} h_{ij}^2,
\end{equation}
we can write
\begin{equation}\label{eq:sigma_matrix}
    S = \begin{bmatrix}
    H_{1}&0&\dots &0 &0&\dots &0\\
    0&H_{2}&\dots &0 &0&\dots &0\\
    \vdots &\vdots &\ddots &\vdots &\vdots&\ddots &\vdots\\
    0 & 0 &\dots &H_m &0 &\dots &0
    \end{bmatrix}.
\end{equation}
Conversely, suppose a family of positive coefficients $\{H_i\}_{i=1}^m$ is given, satisfying the normalization condition
\begin{equation}
\sum_{i=1}^m H_i^2 = n.
\end{equation}
We can then define $S$ in the form~\eqref{eq:sigma_matrix}. The question remains whether there exists an $n\times n$ orthogonal matrix $W$ such that $h = S W$ is ROCN.

Observe first that condition~\eqref{eq:hmatrix_1} is automatically satisfied for any orthogonal $W$. The nontrivial constraint arises from condition~\eqref{eq:hmatrix_2}, i.e., the normalization of the columns of $h$, which reads:
\begin{align}\label{eq:hmatrix2}
    \sum_{i=1}^n h_{ij}^2 = \sum_{i=1}^m \sum_{k,\ell=1}^nS_{ik}W_{k j }S_{i\ell}W_{\ell j}=\sum_{i=1}^m H_i^2 W_{ij}^2=1.
\end{align}
Hence, for $h$ to be ROCN, the squared entries $W_{ij}^2$ must transform the vector of squared row norms $\vec H=(H_1^2,H_2^2,\dots,H_m^2)$ into the uniform vector $\vec{1}_n$. By Horn’s lemma~\cite[Theorem~4]{Horn1954}, since $\vec H$ majorizes $\vec{1}_n$, there always exists an orthogonal matrix $W$ for which condition~\eqref{eq:hmatrix2} holds. However, such a matrix $W$ will not satisfy Eq.~\eqref{eq:hmatrix_3} unless $\vec H$ is uniform. Consequently, the corresponding ROCN Bell operator is not Platonic.

\begin{remark}\label{remark:bound}
For a Platonic Bell inequality, the quantum bound can be obtained from the general result~\cite[Theorem~1]{Epping2013},
\begin{equation}\label{eq:bound_epping}
    \beta_Q^h\leqslant \sqrt{mn} \norm{h}_{2},
\end{equation}
where $\norm{h}_2$ denotes the maximal singular value of $h$. From the singular value decomposition~\eqref{eq:singular_value}, and using the form of $S$ given in Eq.~\eqref{eq:sigma_matrix}, we have
\begin{equation}
    \norm{h}_2 = \max_{i\in\{1,\dots,m\}}H_i.
\end{equation}
Since
\begin{equation}
    \sum_{i=1}^m H_i^2 = \sum_{i=1}^m\sum_{j=1}^n h_{ij}^2 = n,
\end{equation}
it follows that
\begin{equation}
    \max_{i\in\{1,\dots,m\}} \sum_{j=1}^n h_{ij}^2 \geqslant \frac{n}{m},
\end{equation}
with equality if and only if condition~\eqref{eq:hmatrix_3} holds. Therefore Eq.~\eqref{eq:bound_epping} provides the exact quantum bound of an ROCN Bell inequality if and only if it is a Platonic Bell inequality.
\end{remark}
\begin{figure}
    \begin{center}
    \includegraphics[width=\linewidth]{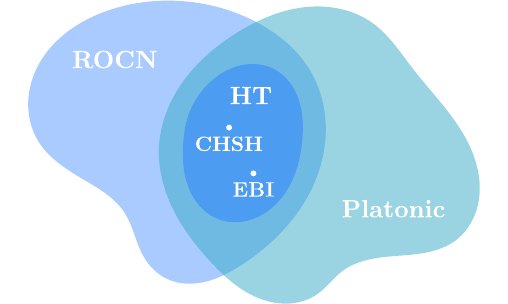}
    \end{center}
    \caption{The three classes of Bell inequalities are compared. The Hadamard truncated (HT) class, introduced in Subsection~\ref{sec:hadamard_matrices}, generalizes both the CHSH and Gisin's elegant Bell inequality (EBI). The ROCN family, discussed in this work, and the Platonic Bell inequalities, introduced in Ref.~\cite{Pal2022}, both contain the Hadamard truncated Bell inequalities.}
\end{figure}
\subsection{A note on Hadamard matrices}
\label{sec:hadamard_matrices}
In this section, we apply the general methods developed in Subsections~\ref{sec:quantum-bound} and~\ref{sec:classical-bound} to a particular class of matrices---namely, the class of real-valued Hadamard matrices (RHM)---thereby extending the results of Ref.~\cite{Goyeneche2023}. These objects form a special subclass of ROCN matrices.

A Hadamard matrix $H$ is defined as a square matrix with entries in $\{-1,1\}$ and orthogonal rows (and consequently columns). Hence, $H$ satisfies the condition $HH^{\mathsf{T}} = H^{\mathsf{T}}H = n \mathbbm{1}_n$.
The dimension of any RHM is necessarily either $n=2$ or $n=4d$ for some $d\in\mathbb{N}$. However, the general existence of such matrices for all $n$ divisible by $4$ remains unresolved---this is the content of the long-standing \emph{Hadamard conjecture}. The smallest case for which the conjecture remains open is $n = 668$~\cite{Horadam2007, DGK14}.

Observe that if $H$ is a Hadamard matrix, then we can remove any number of rows to obtain a new matrix which, up to a normalisation factor, satisfies the ROCN conditions from Eqs.~\eqref{eq:hmatrix_1}--\eqref{eq:hmatrix_2}. As a concrete example, consider the $n = 4$ Hadamard matrix and its modified version after removing the first row:
\begin{equation}
    H=\begin{bmatrix}
    +1&+1&+1&+1\\
    +1&+1&-1&-1\\
    +1&-1&+1&-1\\
    +1&-1&-1&+1\\
    \end{bmatrix} \ \mapsto \
    H'=\begin{bmatrix}
    +1&+1&-1&-1\\
    +1&-1&+1&-1\\
    +1&-1&-1&+1\\
    \end{bmatrix}.\nonumber
\end{equation}
The corresponding normalized ROCN matrix is $h = H'/\sqrt{3}$, which generates the Gisin's elegant Bell inequality in Eq.~\eqref{eq:elegant_ineq}. Interestingly, the construction described above allows one to obtain both examples in which the Bell inequality exhibits quantum advantage ($\beta_C<\beta_Q$) or saturation ($\beta_C=\beta_Q$). In general, given an $n\times n$ real Hadamard matrix $H$, if we remove $n-m$ rows to obtain the truncated Hadamard matrix $H'$, the corresponding ROCN matrix is in the form
\begin{equation}
    h=\frac{H'}{\sqrt m}.
\end{equation}

\begin{remark}
Take a generic Hadamard matrix $H$ of dimension $n = 4d$. By removing a single row, we obtain a \textit{partial} or \textit{truncated Hadamard matrix} $H'$. The corresponding ROCN matrix $h$ for the Bell operator takes the form
\begin{equation}\label{eq:hmatrix_hadamard}
    h=\frac{H'}{\sqrt{4d-1}}.
\end{equation} 
According to Theorem~\ref{th:1}, a Bell operator constructed from $h$ as in Eq.~\eqref{eq:bell_operator} attains maximal quantum value $\beta_Q^h = 4d$. The classical value $\beta_C^h$ is given by
\begin{align}
    \beta_C^{h} &= \max_{\vec a\in\{-1,1\}^{(4d-1)}}\sum_{j=1}^{4d}\abs{\sum_{i=1}^{4d-1}h_{ij} a_i} \nonumber\\
    &= \frac{1}{\sqrt{4d-1}} \underbrace{\left(\max_{\vec a\in\{-1,1\}^{(4d-1)}}\sum_{j=1}^{4d}\abs{\sum_{i=1}^{4d-1}H'_{ij} a_i}\right)}_{=\, C}\leqslant \beta_Q,
\end{align}
where the quantity $C$ is by construction a natural number. We show in this case that $\beta_C^{h}<\beta_Q^h$. Suppose, for contradiction, that $\beta_C^{h}= \beta_Q^{h}$. Then, the equality $C = 4d\sqrt{4d-1}$ must hold, which implies that $\sqrt{4d-1}$ is a rational number of the form
\begin{equation}
    \sqrt{4d-1}=\frac{p}{q}, \ \text{with } p, q \in \mathbb{N} \text{ coprime}.
\end{equation}
Then $4d-1=p^2/q^2$, and since $p, q$ are coprime, we conclude that $q = 1$. Therefore, $\sqrt{4d-1}$ must be an odd integer, i.e., $\sqrt{4d-1} = 2k+1$. Squaring both sides and rearranging, we obtain
\begin{equation}
    d = k^2 + k + \frac{1}{2}\not\in\mathbb{N},
\end{equation}
which contradicts the assumption. Therefore $\beta_C^h<\beta_Q^{h}$.
\end{remark}

In general, when we consider full (untruncated) Hadamard matrices with $m=n$, we can find examples for which the quantum bound is equal to the classical bound, showing that conditions~\eqref{eq:hmatrix_1} and~\eqref{eq:hmatrix_2} are not sufficient to have a non-trivial Bell inequality. In particular, the bound $\beta_C^h=\beta_Q^{h} = n = 4d$ is attained only for \emph{constant row sum} (or \emph{regular}) matrices~\cite{Best1977} such that
\begin{equation}
    \forall i\in \{1,\ldots,n\}: \quad \sum_{j=1}^n h_{ij}={\rm const} = \mp 2\sqrt{d}.
\end{equation}
In other words, the classical value must saturate the Best bound~\cite{Best1977}:
\begin{equation}
n^2 2^{-n}\binom{n}{n/2 }\leqslant \sqrt{n}\beta_C^{h}\leqslant n\sqrt{n}.\label{best-bound}
\end{equation}
Regular matrices exist only for square dimensions, $n=l^2$. The value of $l$ can be easily inferred from the relation $n=4d=l^2$, where $d$ is an integer guaranteed by the sufficient condition on the existence of a Hadamard matrix.

\subsubsection{Classical value and excess of a Hadamard matrix}

Maximal (optimal) local value $\beta_C$ of a given Bell inequality, $
\langle\mathcal{ {B}}\rangle \leqslant \beta_C$, is given by the optimization over all local strategies, which do not allow for quantum correlations. Mathematically, up to normalization, $\beta_C$ corresponds to the value of the so-called \emph{optimized excess} of $H$.

The notion of \emph{excess} was originally defined for Hadamard matrices in the pure mathematical context without any relations to quantum nonlocality~\cite{Schmidt1973, Best1977}. The classical excess is just a sum of all matrix elements
\begin{equation}
\sigma(H) \equiv \sum_{j}\sum_{k}H_{jk},
\end{equation}
while its \emph{optimized} counterpart~\cite{Goyeneche2023}, denoted by $\Sigma$, yields a similar number achieved along the Hadamard orbit which is defined by the monomial transformation: two real Hadamard matrices $H_1$ and $H_2$ of the same size are said to be equivalent, if one can be obtained from another by a monomial transformation $H_1 = M_1 H_2 M_2$, where $M_j=P_j \odot D_j$ is a monomial matrix -- a permutation matrix multiplied (entry-wise Hadamard product) by a diagonal matrix with entries in $\{-1,1\}$.
For real-valued Hadamard matrices, monomial matrices contain only values in $\{-1,1\}$. 
Hence
\begin{equation}
    \Sigma(H)\equiv \max_{M_1,M_2}\Big\{\sum_{j}^{}\sum_{k}^{}\big[M_1HM_2\big]_{jk}\Big\}=\sqrt{n} \beta_{C}^{h}.
\end{equation}
Clearly the excess is permutation-invariant, hence only multiplication by diagonal matrices affects its value. Equivalently, such a multiplication corresponds to the negation of particular rows and columns.

Monomial equivalence relations stratify the set of Hadamard matrices of a given size $n$ onto cosets. For $n=2, 4, 8$ and $12$ there is only a single equivalence class, while for $n\geqslant 16$ the number of classes grows hyper-exponentially~\cite{Craigen1991, Haagerup1996, Orrick2008}:
\[
\begin{tabular}{|c||c|c|c|c|c|c|c|c|c|c|}
  \hline
  $n$ & 2 & 4 & 8 & 12 & 16 & 20 & 24 & 28 & 32 & ... \\
  \hline
  \texttt{\#} of classes & 1 & 1 & 1 & 1 & 5 & 3 & 60 & 487 & 13710027 & ... \\
  \hline
\end{tabular}
\]
This makes the calculation of excess extremely difficult, as it might be that, depending on the representative of the equivalence class, the optimized excess differs. Already for $n=16$ one observes that $\Sigma(H_{16A})=\Sigma(H_{16B})=\Sigma(H_{16C}) = 64$ and $\Sigma(H_{16D})=\Sigma(H_{16E})=56$, where the indices $A$--$E$ denote one of five representatives of a $16$-dimensional Hadamard matrix~\cite{Tadej2006}. Observe that for $H_{16A},H_{16B}$ and $H_{16C}$, the maximal excess equals $64$, which implies that the corresponding maximal classical value equals the maximal quantum value, $\beta_C^h=\beta_Q^h=16$.

When calculating the value of the optimized excess we can utilize several tricks and methods in order to speed-up the numerical procedure. This includes: 
``binarization''~\footnote{Entries of $\mp$-Hadamard matrix as well as the classical strategies can be treated as binary strings and all vector/matrix transformations can be reduced to bit-wise operations.}, parallel computing, and a simple reduction of checking all possible (classical) strategies described in Ref.~\cite{PhysRevA.96.012113} and reproduced in Eq.~\eqref{reduced-trick}. In particular, this last property allows us to consider only a single strategy for one party, say Bob, while the strategy for Alice is fixed to all-one settings---hence, in the following we are going to refer only to this one, non-trivial strategy.

For $n=8$, removing any row yields the optimal excess equal to $18$ and every time the optimal strategy can be chosen to $[1, 1, 1, 1, 1, 1, 1, -1]$. This is one of many strategies that work irrespective of a way the matrix is truncated by cutting a row. However, not all strategies have this property, for example $[1, -1, -1, 1, -1, 1, 1]$ is excluded. For $n=16$ the situation is different. Optimized (maximized) excess for a partial $H_{16}$ is $60$, but there is no single strategy shared between all cases. This shows that one must find such a strategy individually for every removed row (as well as for each equivalence class). Numerical studies suggest the following observation:
\begin{conjecture}
The value of $\beta_C^h$ for the Bell operator induced by $H$ does not depend on the row that is removed.
\end{conjecture}

The conjecture has been verified for all real-valued representatives of Hadamard matrices of size $n\in\{4,8,12,16,20\}$ and partially for some higher dimensions, and as such it requires separate studies.

\section{Conclusions}

In this work, we introduced a general framework for constructing Bell inequalities with an arbitrary number of binary measurements, for which the quantum bound can be determined analytically. This construction is based on the class of row-orthogonal, column-normalized (ROCN) matrices (Definition~\ref{def:ROCN_matrix}). The resulting Bell inequalities, which we refer to as ROCN Bell inequalities, admit an explicit sum-of-squares (SOS) decomposition of the corresponding Bell operators, providing a direct route to compute their quantum bound.

We further derived a necessary and sufficient condition under which an ROCN Bell inequality enables self-testing. In particular, we proved self-testing of families of Clifford observables, which play a fundamental role in quantum theory and quantum information.

Within this framework, we identified a previously unrecognized self-testing equivalence that emerges when certifying an odd number of Clifford observables. In the canonical case of three Pauli matrices, and more generally for representations on an odd number of qubits, this reduces to the well-known complex-conjugation ambiguity. For representations on an even number of qubits, the ambiguity arises from the existence of two unitarily inequivalent irreducible representations of the Clifford algebra, related by an outer automorphism. This transformation plays a role analogous to complex conjugation and can be incorporated within the self-testing formalism. For an even number of Clifford generators, self-testing under the standard definition remains possible.

We also compared ROCN Bell inequalities to other families with analytically accessible quantum bounds. In particular, Platonic Bell inequalities share a similar structural form; however, there exist ROCN inequalities that are not Platonic, and Platonic inequalities that are not ROCN. A notable subclass of ROCN Bell inequalities arises from Hadamard matrices. In this case, the classical bound of the associated Bell inequality coincides exactly with the classical excess of the Hadamard matrix, revealing a direct connection between combinatorial matrix theory and Bell nonlocality.

Several topics remain to be explored. The identification of a previously unrecognized self-testing equivalence motivates the development of a general framework that systematically incorporates such invariances. Another open question is whether a broader class of Bell inequalities can be defined that unifies the ROCN and Platonic families while maintaining analytic computability of the quantum bound. Finally, the framework introduced in this work---particularly the construction of Bell inequalities and the associated self-testing statements---appears to extend naturally to the entanglement-swapping network scenario~\cite{PhysRevA.85.032119}, potentially enabling the derivation of Bell inequalities for detecting nonbilocality and self-testing the quantum network~\cite{Autor2025}.

\section{Acknowledgments}
We are indebted to Robert Craigen for his comments on partial Hadamard matrices. This work is supported by the National Science Centre (Poland) through the SONATA BIS project No.\ 019/34/E/ST2/00369, and funding from the European Union's Horizon Europe research and innovation programme under grant agreement No.\ 101080086 NeQST.

\appendix

\section{Jordan--Wigner representation}
\label{app:ST-theorems}

In this appendix, we present results used in the proof of self-testing. We begin with the proof of the Jordan--Wigner representation theorem.

\begin{proof}[Proof of Theorem~\ref{th:1}]
    If $m=1$, the result is trivial. Let $m\geqslant 2$ and set $r=\lfloor m/2\rfloor \geqslant 1$. Consider the first pair of observables $A_1$ and $A_2$ satisfying
    \begin{equation}\label{eq:constraints}
        \{A_1, A_2\} = 0, \quad A_1^2 = A_2^2 = \mathbbm{1}.
    \end{equation}
    Using the cyclic property of the trace together with the above relations, we find:
    \begin{equation}
        \Tr[A_1] = \Tr[A_2 A_2 A_1] = \Tr[A_2 A_1 A_2] = - \Tr[A_1],
    \end{equation}
    which implies that $A_1$ is traceless. Since the eigenvalues of $A_1$ are $\pm 1$, there exists a unitary operator $\tilde{U}_1$ such that
    \begin{equation}
        \tilde{U}_1 A_1 \tilde{U}_1^\dag = Z \otimes \mathbbm{1}.
    \end{equation}
    We now decompose the transformed operator $\tilde{U}_1 A_2 \tilde{U}_1^\dag$ as
    \begin{equation}\label{eq:decomposition}
        \tilde{U}_1 A_2 \tilde{U}_1^\dag = \mathbbm{1} \otimes A'_\mathbbm{1} + X \otimes A'_X + Y \otimes A'_Y + Z \otimes A'_Z.
    \end{equation}
    From Eq.~\eqref{eq:constraints}, it follows that
    \begin{align}
        A'_\mathbbm{1} = A'_Z = 0,\quad        {A'_X}^2 + {A'_Y}^2 = \mathbbm{1},\quad
        [A_X', A_Y'] = 0.
    \end{align}
    Since $A_X'$ and $A_Y'$ commute, they can be expressed in a common eigenbasis $\{\ket{k}\}$, thus yielding
    \begin{align}
        \tilde{U}_1 A_2 \tilde{U}_1^\dag &= \sum_k \left[X \cos\psi_k - Y \sin\psi_k\right] \otimes \ketbra{k} \nonumber\\ 
        &= \sum_k \begin{bmatrix}
        0 & e^{\mathrm{i}\psi_k}\\
        e^{-\mathrm{i}\psi_k} & 0 
        \end{bmatrix} \otimes \ketbra{k}.
    \end{align}
    Define a unitary operator $\tilde{V}$ as
    \begin{equation}
        \tilde{V} = \sum_k \begin{bmatrix}
        1 & 0 \\
        0 & e^{\mathrm{i}\psi_k}
        \end{bmatrix} \otimes \ketbra{k}.
    \end{equation}
    Then, $\tilde{V} \tilde{U}_1 A_2 \tilde{U}_1^\dag \tilde{V}^\dag = X \otimes \mathbbm{1}$ and $[\tilde{V}, Z \otimes \mathbbm{1}] = 0$. Setting $U_1 = \tilde{V} \tilde{U}_1$, we obtain
    \begin{equation}
        U_1 A_1 U_1^\dag = Z \otimes \mathbbm{1}, \quad U_1 A_2 U_1^\dag = X \otimes \mathbbm{1}.
    \end{equation}
    From the canonical anticommutation relations~\eqref{eq:anticomm_th1}, it follows that the remaining observables $U_1 A_i U_1^\dag$ for $i \in \{3,\ldots,m\}$ are of the form
    \begin{equation}
        U_1 A_i U_1^\dag = Y \otimes A'_i,
    \end{equation}
    where, for all $i,k\in\{3,\dots,m\}$, the Hermitian operators $A'_i$ and $A'_k$ satisfy
    \begin{equation}
        \{A'_i, A'_k\} = 2\delta_{ik} \mathbbm{1}.
    \end{equation}
    We now repeat the same construction for $A'_3$ and $A'_4$, obtaining a unitary operator $U_2$ such that
    \begin{align}
        &(\mathbbm{1}_2 \otimes U_2) U_1 A_3 U_1^\dag (\mathbbm{1}_2 \otimes U_2)^\dag = Y \otimes Z \otimes \mathbbm{1}, \\
        &(\mathbbm{1}_2 \otimes U_2) U_1 A_4 U_1^\dag (\mathbbm{1}_2 \otimes U_2)^\dag = Y \otimes X \otimes \mathbbm{1}.
    \end{align}
    Iterating this procedure $r$ times yields the representation introduced in Eq.~\eqref{eq:U}, with the overall unitary transformation
    \begin{equation}
        U = (\mathbbm{1}_2^{\otimes(r-1)} \otimes U_r)(\mathbbm{1}_2^{\otimes(r-2)} \otimes U_{r-1}) \dots (\mathbbm{1}_2 \otimes U_2) U_1.
    \end{equation}
    For odd $m$, the unpaired element $A_m$ takes, after applying $U$, the form given in Eq.~\eqref{eq:last_operator_JW}, with $A_m^{'2}=\mathbbm{1}$.
\end{proof}

The following result is a consequence of the Jordan–Wigner representation theorem and characterizes the structure of its commutant.
\begin{theorem}\label{th:3}
    Let $\{A_i\}_{i=1}^m$ be a set of observables acting on a finite-dimensional Hilbert space $\mathcal{H}$ and satisfying the canonical anticommutation relations~\eqref{eq:anticomm_th1}. Let $m=2r+\varepsilon$, with $\varepsilon\in\{0, 1\}$ and $r=\lfloor m/2\rfloor$, and let $P$ be a positive operator satisfying $[P, A_i] = 0$ for all $i \in \{1,\ldots,m\}$. Then, with respect to the Jordan--Wigner representation given in Eqs.~\eqref{eq:U} and~\eqref{eq:last_operator_JW}, the operator $P$ can be written as
    \begin{equation}\label{eq:commuting_operator}
        P= \mathbbm{1}_2^{\otimes r} \otimes P',
    \end{equation}
    where $P'$ is a positive operator. For $\varepsilon=1$ (i.e., odd $m$), we also have
    \begin{equation}\label{eq:commutation_last_observable_JW}
        [P',A'_m]=0.
    \end{equation}
\end{theorem}
\begin{proof}
    We prove the claim by iteratively decomposing $P$ in the Pauli basis according to the Jordan--Wigner representation. Start with the first two observables $A_1$ and $A_2$, which act non-trivially on the first $\mathbb{C}^2$ subspace. Then $P$ can be expressed as
    \begin{equation}
        P = \mathbbm{1}_2 \otimes P'_\mathbbm{1} + X \otimes P'_X + Y \otimes P'_Y + Z \otimes P'_Z.
    \end{equation}
    The commutation relations $[P, A_1] = [P, A_2] = 0$ imply that $P'_X = P'_Y = P'_Z = 0$, so that
    \begin{equation}
        P = \mathbbm{1}_2 \otimes P'_\mathbbm{1},
    \end{equation}
    with $P'_\mathbbm{1}$ being a positive operator. Repeating this procedure iteratively, we obtain Eq.~\eqref{eq:commuting_operator}. Finally, for odd $m$, the last unpaired observable $A_m$ acts on the remaining factor, which gives Eq.~\eqref{eq:commutation_last_observable_JW}.
\end{proof}
An alternative proof is obtained by noting that the first $2r$ observables in Eq.~\eqref{eq:U} generate the full algebra $\mathcal B(\mathbb{C}_2^{\otimes r})\otimes \mathbbm{1}$. Any operator commuting with all these generators must therefore be of the form consistent with Eq.~\eqref{eq:commuting_operator}.

\section{Proof of Theorem~\ref{th:self_testing_statement}}\label{app:proof_main_theorem}

In this appendix, we present the proof of Theorem~\ref{th:self_testing_statement} concerning the self-testing statement. We begin with the following preparatory lemma, which establishes that the linear system~\eqref{eq:linear_system} admits a unique solution consisting of pairwise anticommuting observables if and only if the system has full column rank.

\begin{lemma}\label{th:2}
    Let $\{A_i\}_{i=1}^m$ be a family of observables acting on a finite-dimensional Hilbert space $\mathcal{H}$ such that $A_i^2 = \mathbbm{1}$ for all $i \in \{1,\ldots,m\}$. Additionally, let $M$ be an $n \times \frac{1}{2} m (m-1)$ real matrix, and consider the following system of linear equations
    \begin{equation}\label{eq:eqnset_2}
        \forall j \in \{1,\ldots,n\}: \quad \sum_{i < k}^m M_{j,(i,k)} \{A_i, A_k\} = 0,
    \end{equation}
    where each pair $(i,k)$ for $i < k$ is treated as a single index. Then, a set of anticommuting observables~\eqref{eq:anticomm_th1} is the unique solution of~\eqref{eq:eqnset_2} if and only if $M$ has full column rank.
\end{lemma}
\begin{proof}
    Assume, for contradiction, that $M$ does not have full column rank. Then, there exists a nonzero vector $\vec{s} \in \mathbb{R}^{\frac{1}{2} m (m-1)}$ such that $M \vec{s} = 0$. We construct a one-parameter family of operator tuples $\{A_i(\alpha)\}_{i=1}^m$ satisfying Eq.~\eqref{eq:eqnset_2} and which do not satisfy the anticommutation relations~\eqref{eq:anticomm_th1}.
    
    Let $\{\tilde A_i\}_{i=1}^m$ be the family of anticommuting observables as given in the reference strategy (Definition~\ref{def:reference_strategy}). Furthermore, let us define the symmetric $m\times m$ matrix $S$, whose elements are $S_{ii}=0$, and $S_{ki}=S_{ik}=s_{(i,k)}$ for $i < k$. For all real parameters $\alpha$, define
    \begin{equation}
    	G(\alpha) \vcentcolon= G(0) + \alpha S =\mathbbm{1}_m + \alpha S.
    \end{equation}
	For $|\alpha| < 1/\|S\|$, where $\|S\|$ denotes the spectral norm of $S$, the matrix $G(\alpha)$ is positive definite with diagonal entries $1$. In this case, it is a Gram matrix, $G_{ij}(\alpha)=\vec v_i(\alpha)\cdot \vec v_j(\alpha)$ of real unit vectors $\{\vec{v}_1(\alpha),\ldots,\vec{v}_m(\alpha)\} \subset \mathbb{R}^m$. We further define
    \begin{equation}
        A_i(\alpha) = \sum_{j = 1}^m v_{i,j}(\alpha)\tilde A_j.
    \end{equation}	
	The operators $A_i(\alpha)$ are clearly Hermitian, and furthermore
    \begin{align}
        &\{A_i(\alpha), A_k(\alpha)\} = 2 \left[ \vec{v}_i(\alpha) \cdot \vec{v}_k(\alpha) \right] \mathbbm{1}, \\
        &A_i(\alpha)^2 = \| \vec{v}_i(\alpha) \|^2 \mathbbm{1} = \mathbbm{1}.
    \end{align}
	The family of operators $\{A_i(\alpha)\}$ satisfies Eq.~\eqref{eq:eqnset_2} for every $\alpha<1/\|S\|$:
    \begin{align}
        \sum_{i < k}^m M_{j,(i,k)} G_{ik}(\alpha) &= \sum_{i < k}^m M_{j,(i,k)} (0 + \alpha s_{(i,k)}) \nonumber\\
        &= \alpha \left( M \vec{s} \right)_j = 0.
    \end{align}
	Finally, for $\alpha \neq \alpha'$ the Gram matrices $G(\alpha)$ and $G(\alpha')$ are different. Since unitary conjugations correspond to orthogonal transformations on the vectors $\vec v_i$, which preserve Gram matrices, the tuples $\{A_i(\alpha)\}$ and $\{A_i(\alpha')\}$ cannot be unitarily equivalent. This contradicts the uniqueness condition, so $M$ must have full column rank.
\end{proof}

With this preliminary result established, we can now proceed to the proof of the theorem concerning the self-testing statement.

\begin{proof}[Proof of Theorem~\ref{th:self_testing_statement}]
We prove that condition~\eqref{eq:necessary_self_testing} implies that the maximal violation of the Bell inequality self-tests the reference strategy. The converse direction follows from Lemma~\ref{th:2}. The proof is organized into three parts. In \emph{Part~I}, we show that achieving the maximal quantum violation forces Alice’s observables to form a pairwise anticommuting family. In \emph{Part~II}, we employ vectorization techniques to establish the relation between Alice’s and Bob’s observables, as well as the Schmidt coefficients of the shared quantum state. Finally, in \emph{Part~III}, we invoke the Jordan--Wigner representation theorem to demonstrate self-testing of the reference strategy. The distinction between the even and odd $m$ cases becomes explicit in this final step.\\

\emph{Part I.} Let $h$ be an $m\times n$ ROCN matrix, and let $ {\mathcal B}_h$ denote the corresponding Bell operator. Suppose that a bipartite quantum state $\ket{\psi}\in\mathcal{H}_\mathrm{A} \otimes\mathcal{H}_\mathrm{B}$, together with the local observables $\{A_i\}_{i=1}^m$, $\{B_j\}_{j=1}^n$, achieves the maximal quantum violation of the Bell inequality. We assume that all observables $A_i$ and $B_j$ square to identity, and that the shared state $\ket{\psi}$ is pure---an assumption that can always be made~\cite{Baptista2023}. As discussed in Appendix~\ref{app:weird_relation}, we can additionally assume that the reduced density matrices $\rho_{\mathrm{A}}$ and $\rho_\mathrm{B}$ are full rank.

In the proof of Proposition~\ref{prop:quantum_bound}, we introduced an SOS decomposition of $ {\mathcal B}_h$ [see Eq.~\eqref{eq:sos}]. From this, it follows that the maximal quantum value $\beta_Q^{h} = n$ is attained if and only if the shared state $\ket{\psi}$ satisfies $N_j\ket{\psi}=0$ for all $j\in\{1,\dots, n\}$, or equivalently, if the saturation condition~\eqref{eq:saturation_quantum_bound} holds. As a consequence:
\begin{equation}\label{eq:eqnset_0}
    \forall j\in\{1,\ldots,n\}: \quad \Bigr(\sum_{i = 1}^m h_{ij} A_i \otimes B_j \Bigr)^2 \ket{\psi} = \ket{\psi}.
\end{equation}
Expanding the square and using $A_i^2 = B_j^2 = \mathbbm{1}$, we obtain
\begin{equation}\label{eq:linear_system}
\forall j \in \{1,\ldots,n\}: \quad \sum_{i < k}^m h_{ij}h_{kj} \{A_i, A_k\}\ket{\psi} = 0.
\end{equation}
Since the reduced density matrix $\rho_\mathrm{A} = \Tr_\mathrm{B}[\ketbra{\psi}]$ is full rank, it further implies that the system of equations~\eqref{eq:eqnset} is satisfied. If the matrix $M$ defined in Eq.~\eqref{eq:matrix_M} has full column rank, then Eq.~\eqref{eq:eqnset} enforces that Alice’s observables $\{A_i\}$ satisfy the canonical anticommutation relations~\eqref{eq:anticomm_th1}. \\

\textit{Part II.} Let us now express the shared state $\ket{\psi}$ in its Schmidt decomposition:
\begin{equation}
    \ket{\psi} = \sum_{i=1}^D \lambda_i \ket{i} \otimes \ket{i} =  (P_\mathrm{A} \otimes \mathbbm{1}) \ket{\Phi_D},
\end{equation}
where $P_\mathrm{A} = \sqrt{D} \sum_{i=1}^D \lambda_i \ketbra{i}$, $D=\dim \mathcal{H}_{\mathrm{A}}=\dim\mathcal{H}_{\mathrm{B}}$ and
\begin{equation}
    \ket{\Phi_D} =\frac{1}{\sqrt{D}} \sum_{i=1}^D \ket{i}\otimes\ket{i}
\end{equation}
is the maximal entangled state between $\mathcal{H}_\mathrm{A}A$ and $\mathcal{H}_\mathrm{B}$. Observe that, since we assumed the reduced density matrices $\rho_{\mathrm{A}}$ and $\rho_{\mathrm{B}}$ to be full rank, the Hilbert spaces $\mathcal{H}_{\mathrm{A}}$ and $\mathcal{H}_{\mathrm{B}}$ have the same dimensions. Moreover, we choose the Schmidt basis so that it coincides with the eigenbasis of both $\rho_{\mathrm{A}}$ and $\rho_{\mathrm{B}}$.

For convenience, we introduce the operators
\begin{equation}\label{eq:definition_A_bar}
   F_j \vcentcolon = \sum_{i=1}^m h_{ij} A_i,
\end{equation}
so that the saturation condition~\eqref{eq:saturation_quantum_bound} can be rewritten as
\begin{equation}
    \forall j \in \{1,\ldots,n\}: \quad \left(F_j P_\mathrm{A} \otimes B_j \right) \ket{\Phi_D} = (P_\mathrm{A} \otimes \mathbbm{1}) \ket{\Phi_D}.
\end{equation}
Since the anticommutation relations imply that $F_j^2 = \mathbbm{1}$ for all $j \in \{1,\ldots,n\}$, we can apply the operator $F_j \otimes \mathbbm{1}$ to both sides and use the flip identity from Eq.~\eqref{eq:flip-flop} to obtain
\begin{equation}
    \left(\mathbbm{1} \otimes  B_j  P_\mathrm{A} \right) \ket{\Phi_D} = (\mathbbm{1} \otimes P_\mathrm{A} F_j^\mathsf{T})\ket{\Phi_D}.
\end{equation}
Projecting both sides onto $\ket{\ell} \otimes \ket{k}$ for all $k,\ell \in \{1, \ldots, D\}$ yields
\begin{equation}
    \mel{k}{  B_j  P_\mathrm{A} }{\ell} = \mel{k}{ P_\mathrm{A} F_j^\mathsf{T}}{\ell},
\end{equation}
and hence, on the subspace available to Bob
\begin{equation}\label{eq:operators_bob}
    B_j P_\mathrm{A} = P_\mathrm{A} F_j^\mathsf{T}.
\end{equation}
Taking the Hermitian conjugate gives
\begin{equation}
    P_\mathrm{A}  B_j  = F_j^\mathsf{T} P_\mathrm{A}.
\end{equation}
Combining these two equalities, we find
\begin{equation}
     P_\mathrm{A}  B_j B_j  P_\mathrm{A} = P_\mathrm{A}^2 = F_j^\mathsf{T} P_\mathrm{A}^2 F_j^\mathsf{T},
\end{equation}
which implies that $[P_\mathrm{A}^2,  F_j^\mathsf{T}] = 0$. Since $P_\mathrm{A}$ is positive, it follows that $[P_\mathrm{A},  F_j^\mathsf{T}]=0$. From Eq.~\eqref{eq:operators_bob}, we obtain
\begin{equation}
    B_j  = F_j^\mathsf{T} = \sum_{i=1}^m h_{ij} A_i^\mathsf{T}.
\end{equation}
This shows that the form of Bob's observables is determined up to a local isometry. Finally, substituting the definition~\eqref{eq:definition_A_bar} into the commutation relation above gives
\begin{equation}
    \sum_{i=1}^m h_{ij} [P_\mathrm{A}, {A}_i^T]=0.
\end{equation}
Multiplying this equation by $h_{kj}$ and summing over $j$ leads to
\begin{multline}
    \forall k\in\{1,\dots,m \}:\quad \sum_{j=1}^n h_{kj}^2\  [P_\mathrm{A}, {A}_k^\mathsf{T}]=0
    \\\implies [P_\mathrm{A}, {A}_k^\mathsf{T}]=0.
\end{multline}
Altogether, we have established that: Alice's observables $\{A_i\}$ satisfy the canonical anticommutation relations~\eqref{eq:anticomm_th1}, Bob's observables $\{B_j\}$ can be expressed as linear combinations of the anticommuting operators $\{A_i^\mathsf{T}\}$, and the operator $P_{\mathrm{A}}$ commutes with all $A_i$.\\

\textit{Part III.} By the Jordan–Wigner representation theorem, up to a unitary transformation, Alice's Hilbert space can be decomposed as
\begin{equation}
    \mathcal{H}_{\mathrm{A}} = \bigotimes_{k=1}^r \mathbb{C}^2 \otimes \mathcal{H}',
\end{equation}
and the observables take the canonical form
\begin{equation}
    \forall i \in \{1,\ldots,2r\}:\quad A_i = \tilde{A}_i \otimes \mathbbm{1}_{\mathcal{H}'},
\end{equation}
while for the case of odd $m = 2r + 1$,
\begin{equation}
A_m = \tilde{A}_m \otimes A'_m,
\end{equation}
where $A'_m$ is Hermitian and unitary. Furthermore, if we denote $d' = \dim \mathcal{H}'$, then the total Schmidt rank of the shared state is $D= d d'$, with $d=2^r$. Since $P_\mathrm{A}$ commutes with all $A_i$, by Theorem~\ref{th:3} it must take a tensor-product form:
\begin{equation}
    P_\mathrm{A} = \mathbbm{1}_2^{\otimes r} \otimes P'= \mathbbm{1}_2^{\otimes r} \otimes \left(\sum_{\alpha=1}^{d'}{\lambda_\alpha} \sqrt{d'}\ketbra{\alpha}\right).
\end{equation}
Because $[A_m',P']=0$, we can choose the basis $\{\ket{\alpha}\}$ to be a common eigenbasis of $A_m'$ and $P'$, with eigenvalues $\pm 1$. We can then decompose Bob’s Hilbert space as
\begin{equation}
    \mathcal{H}_{\mathrm{B}}=\bigotimes_{k=1}^r \mathbb{C}^2 \otimes \mathcal{H}',
\end{equation}
and the state $\ket{\psi}$ becomes:
\begin{align}
    \ket{\psi}&=\left(\frac{1}{\sqrt{d}}\sum_{\boldsymbol i\in\{0,1\}^r} \ket{\boldsymbol i}\otimes\ket{\boldsymbol{i}}\right)\otimes \left(\sum_{\alpha=1}^{d'}{\lambda_\alpha} \ket{\alpha}\otimes\ket{\alpha}\right) \nonumber\\
    &=\ket{\Phi_d}\otimes \ket{\xi}.
\end{align}
For even $m$, Bob's observables can be written as
\begin{align}
    B_j=\sum_{i=1}^m h_{ij} A_i^{\mathsf{T}} =\sum_{i=1}^{2r} h_{ij} \tilde A_i^{\mathsf{T}}\otimes \mathbbm{1}_{\mathcal{H}'} = \tilde B_j \otimes \mathbbm{1}_{\mathcal{H}'},
\end{align}
which yields the standard self-testing (Definition~\ref{def:self-testing}). For odd $m$, we have
\begin{equation}
    B_j=\sum_{i=1}^m h_{ij} A_i^{\mathsf{T}} =\sum_{i=1}^{2r} h_{ij} \tilde A_i^{\mathsf{T}}\otimes \mathbbm{1}_\mathcal{H'} + h_{mj} \tilde A_m^{\mathsf{T}} \otimes A_m'.
\end{equation}
Let $M_0=N_0$ and $M_1=N_1$ denote the projections onto the $+1$ and $-1$ eigenspaces of $A_m'$ on Alice's and Bob's side, respectively. It follows that
\begin{equation}
	\mel{\xi}{(M_0\otimes N_0+M_1\otimes N_1)}{\xi}=1,
\end{equation}
and, for all $i\in\{1,\dots,m\}$ and $j\in\{1,\dots, n\}$,
\begin{align}
	A_i&=\tilde A_i\otimes M_0+ \tilde A_i^{\whitestar} \otimes M_1,\\
	B_j&=\tilde B_i\otimes N_0+ \tilde B_j^{\whitestar} \otimes N_1.
\end{align}
These relations show that in the case of odd $m$, the observables are self-tested in the sense of Eq.~\eqref{eq:two_parties_self-testing}.
\end{proof}

\section{Lifting the full rank assumption}
\label{app:weird_relation}

In self-testing, one should avoid making any a priori assumptions about the \emph{test} strategy that reproduces the observed correlations. In full generality, this means allowing the shared state to be mixed, the reduced density matrices to be non-full rank, and the measurements to be general POVMs. In our framework, the latter corresponds to relaxing the condition $A_i^2 = B_j^2 = \mathbbm{1}$ to $A_i^2, B_j^2 \leqslant \mathbbm{1}$. However, by invoking Naimark’s dilation theorem together with purification, one can, without loss of generality, restrict attention to pure states and observables that square to the identity (see Ref.~\cite[Theorem 4.1]{Baptista2023}). We now show that, in addition, one may assume without loss of generality that the reduced density matrix on either subsystem is full rank.
\begin{theorem}
    To establish self-testing from the maximal violation of an ROCN Bell inequality, it is sufficient to assume that the observables ${A_i}$ and ${B_j}$ square to the identity, that the shared state $\ket{\psi}$ is pure, and that the reduced density matrices $\rho_\mathrm{A} = \Tr_\mathrm{B}[\ketbra{\psi}]$ and $\rho_\mathrm{B} = \Tr_\mathrm{A}[\ketbra{\psi}]$ are full rank.
\end{theorem}
\begin{proof}
It suffices to show that one can restrict the Hilbert spaces $\mathcal{H}_{\mathrm{A}}$ and $\mathcal{H}_\mathrm{B} $ to the supports of $\rho_\mathrm{A}$ and $\rho_\mathrm{B}$, respectively, without affecting the properties of the observables, in particular ensuring that $A_i^2 = \mathbbm{1}$ still holds on the reduced supports. Consider the operators
\begin{equation}\label{eq:def_tilde A_j}
    F_j \vcentcolon= \sum_{i=1}^m h_{ij}A_i.
\end{equation}
Suppose that the state $\ket{\psi}$ achieves the maximal quantum violation. Then it satisfies the saturation condition~\eqref{eq:saturation_quantum_bound}, which can be rewritten as
\begin{multline}\label{eq:condition_FB}
    \forall j \in\{1,\dots,n\}:\quad F_j \otimes B_j\ket{\psi}=\ket{\psi} \\
    \implies F_j \otimes \mathbbm{1}\ket{\psi}=\mathbbm{1} \otimes B_j\ket{\psi}.
\end{multline}
We begin with Alice’s subsystem. Let $\mathcal{H}_0$ denote the support of $\rho_\mathrm{A}$, and define its orthogonal complement $\mathcal{H}_1 = \mathcal{H}_0^\perp$, so that the full Hilbert space decomposes as
\begin{equation}\label{eq:decomposition_self-testing}
    \mathcal{H}_\mathrm{A} =\mathcal{H}_0\oplus \mathcal{H}_1.
\end{equation}
Let $R_0$ and $R_1$ be the orthogonal projections onto $\mathcal{H}_0$ and $\mathcal{H}_1$, respectively. Clearly, the state $\ket{\psi}$ is supported entirely on $\mathcal{H}_0$: 
\begin{align}
    R_0\otimes \mathbbm{1} \ket{\psi}=\ket{\psi},\quad 
    R_1\otimes \mathbbm{1} \ket{\psi}=0.
\end{align}
Next, we define the block components of $F_j$ with respect to the decomposition~\eqref{eq:decomposition_self-testing}.
For $\alpha,\beta \in \{0,1\}$, set $F_j^{(\alpha\beta)}=R_{\alpha} F_j R_{\beta}$, so that
\begin{equation}
    F_j=\begin{bmatrix}
        F_j^{(00)} & F_j^{(01)}\\
        F_j^{(10)} & F_j^{(11)}
    \end{bmatrix}.
\end{equation}
We have
\begin{align}\label{eq:null_operators}
    R_1 F_j R_0\, \rho_\mathrm{A} &= \Tr_\mathrm{B}[R_1 F_j R_0\otimes \mathbbm{1}\ketbra{\psi}] \nonumber\\
    &=\Tr_\mathrm{B}[R_1 F_j \otimes \mathbbm{1}\ketbra{\psi}] \nonumber\\
    &=\Tr_\mathrm{B}[R_1\otimes B_j\ketbra{\psi}]=0.
\end{align}
It follows that $R_1 F_j R_0=0$, and by Hermiticity of $F_j$, also $R_1 F_j R_0=0$. Therefore, each operator $F_j$ is block-diagonal with respect to the decomposition $\mathcal{H}_\mathrm{A}=\mathcal{H}_0\oplus \mathcal{H}_1$:
\begin{equation}
    F_j=\begin{bmatrix}
        F_j^{(00)} & 0\\
        0 & F_j^{(11)}
    \end{bmatrix}.
\end{equation}
From the definition of $F_j$ in Eq.~\eqref{eq:def_tilde A_j}, and since $h$ has orthogonal (and hence linearly independent) rows, it follows that
\begin{equation}
\forall i \in\{1,\dots,m\}:\quad R_1 A_i R_0 = 0.
\end{equation}
By Hermiticity of $A_i$, we also have $R_0 A_i R_1 = 0$. Therefore each observable $A_i$ is also block-diagonal with respect to the decomposition $\mathcal{H}_\mathrm{A}=\mathcal{H}_0\oplus \mathcal{H}_1$:
\begin{equation}
A_i = \begin{bmatrix}
A_i^{(00)} & 0 \\
0 & A_i^{(11)}
\end{bmatrix}.
\end{equation}
It is then clear that the restricted operators $A_i^{(00)}$ square to the identity on $\mathcal{H}_0$. Hence, we may restrict the observables $A_i$ to the subspace $\mathcal{H}_0$ without loss of generality. The same argument applies to Bob’s observables. Specifically, using Eq.~\eqref{eq:condition_FB} and repeating the reasoning in~\eqref{eq:null_operators}, one concludes that the observables $\{B_j\}$ can be restricted to the support of $\rho_\mathrm{B}$ while preserving $B_j^2 = \mathbbm{1}$.
\end{proof}

\section{Partial transposition vs. local unitary and complex conjugation equivalences}
\label{app:partial_transposition_vs_cc_u}
In this appendix, we show explicitly that the two strategies $\{\tilde A_i\},\  \{\tilde B_j\},\  \ket{\Phi_d}$ and $\{{\tilde A}^{\whitestar}_i\},\  \{{\tilde B}^{\whitestar}_j\},\  \ket{\Phi_d}$, introduced in Definition~\ref{def:reference_strategy} and Eqs.~\eqref{eq:A_whitestar}--\eqref{eq:B_whitestar}, cannot, in general, be obtained from one another by combining complex conjugation with local unitary transformations. We focus on the case where $m=2r +1$ is odd. 

Assume, toward a contradiction, that there exist unitary operators $U_\mathrm{A}$ on $\mathcal{H}_\mathrm{A}$ and $U_{\mathrm{B}}$ on $\mathcal{H}_\mathrm{B}$ such that, for all $i,j$,
\begin{align}\label{eq:strategy_canonical_app}
    {\tilde A}^{\whitestar}_i&=U_\mathrm{A} \tilde A_i U_{\mathrm{A}}^\dagger,\\
    {\tilde B}^{\whitestar}_j&=U_\mathrm{B} \tilde B_j U_{\mathrm{B}}^\dagger,\\
    \ket{\Phi_d}&=U_{\mathrm{A}}\otimes U_{\mathrm{B}}\ket{\Phi_d}. \label{eq:strategy_partial_transposed_app}
\end{align}
By construction,
\begin{equation}
    \forall i\in\{1,\dots,2r\}:\quad \tilde A^{\whitestar}_i=\tilde A_i,
\end{equation}
and consequently, the operator $U_\mathrm{A}$ must commute with every $\tilde A_1,\dots,\tilde A_{2r}$. Since the $C^*$-algebra generated by $\{A_i\}_{i=1}^{2r}$ is the full operator algebra $\mathcal B(\mathbb{C}^{2\otimes r})$, it follows that $U_\mathrm{A}$ must be proportional to the identity. Consequently, the final condition
\begin{equation}
    \tilde A_{2r+1}^{\whitestar}=-\tilde A_{2r+1}=U_{\mathrm{A}}\tilde A_{2r+1} U_{\mathrm{A}}^\dagger
\end{equation}
cannot be satisfied. 

Next, consider the complex conjugates $\tilde A_{i}^*$ and $\tilde B_j^*$. Suppose we seek local unitaries such that, for all $i,j$,
\begin{align}
    {\tilde A}^{\whitestar}_i&=U_\mathrm{A} \tilde A_i^* U_{\mathrm{A}}^\dagger,\\
    {\tilde B}^{\whitestar}_j&=U_\mathrm{B} \tilde B_j^* U_{\mathrm{B}}^\dagger,\\
    \ket{\Phi_d}&=U_{\mathrm{A}}\otimes U_{\mathrm{B}}\ket{\Phi_d}. 
\end{align}
Using the explicit representation of $\tilde A_i$ given in Theorem~\ref{th:1}, we have
\begin{equation}
    \forall i\in\{1,\dots, m\}:\quad \tilde A_i^{*}=(-1)^{\lfloor \frac{i-1}{2}\rfloor}\tilde A_i.
\end{equation}
In particular, $\tilde A_m^*=(-1)^r \tilde A_m^*$.
Consider the operator
\begin{equation}
    U_\mathrm{A}= \mathbbm{1}\otimes Y \otimes \mathbbm{1}\otimes Y\otimes \dots,
\end{equation}
and notice that
\begin{equation}
    \forall i\in\{1,\dots,2r\}:\quad U_\mathrm{A}\tilde A_i U_\mathrm{A}^\dagger =(-1)^{\lfloor \frac{i-1}{2}\rfloor} \tilde A_i,
\end{equation}
while $U_\mathrm{A} \tilde A_m U_\mathrm{A}^\dagger=\tilde A_m$. 

As a consequence, when $r$ is even (corresponding to an even number of qubits), we obtain
\begin{equation}
    \forall i\in\{1,\dots,m\}: \quad \tilde A_i=U_{\mathrm{A}}\tilde A_i^*U_{\mathrm{A}}^\dagger.
\end{equation}
In this case, the operators $\tilde A_i^*$ are unitarily equivalent to $A_i$, and from the first part of the proof it follows that the two strategies~\eqref{eq:strategy_canonical_app} and~\eqref{eq:strategy_partial_transposed_app} cannot be transformed into each other via complex conjugation and unitaries. More general transformation should be considered~\cite{konderak_asymptotics}.

For $r$ odd (corresponding to an odd number of qubits), we find instead
\begin{equation}
    \forall i\in\{1,\dots,m\}: \quad \tilde A_i^{\whitestar}=U_{\mathrm{A}}\tilde A_i^*U_{\mathrm{A}}^\dagger.
\end{equation}
By choosing $U_{\mathrm{B}}=U^*_{\mathrm{A}}=U^{\mathsf{T}}_{\mathrm{A}}$, it follows that
\begin{equation}
    \forall j\in\{1,\dots,n\}:\quad \tilde B_j^{\whitestar}=U_{\mathrm{B}}\tilde B_j^*U_{\mathrm{B}}^\dagger,
\end{equation}
and $\ket{\Phi_d}=U_{\mathrm{A}}\otimes U_\mathrm{B}\ket{\Phi_d}$. Hence, in this case, the two strategies~\eqref{eq:strategy_canonical_app} and~\eqref{eq:strategy_partial_transposed_app} can indeed be related via complex conjugation and local unitary transformations. 

\bibliography{bibliography}
\end{document}